\documentclass[a4paper,12pt]{article}

\usepackage[english]{babel}
\usepackage[T1]{fontenc}
\usepackage[utf8]{inputenc}

\usepackage[margin=2.5cm]{geometry}
\usepackage{setspace}
\usepackage[tt=false,type1=true]{libertine}
\usepackage[varqu]{zi4}
\usepackage[libertine]{newtxmath}

\usepackage[dvipsnames]{xcolor}
\definecolor{linkblue}{RGB}{25,70,140}

\usepackage{booktabs}
\usepackage{multirow}
\usepackage{array}
\usepackage{tabularx}
\usepackage{makecell}

\usepackage{algorithm}
\usepackage{algpseudocode}

\usepackage{amsthm}
\usepackage{mathtools}
\usepackage{textcomp}

\usepackage{graphicx}
\usepackage{tikz}
\usetikzlibrary{
    positioning,
    arrows.meta,
    decorations.pathreplacing,
    shapes.geometric
}

\usepackage{enumitem}
\setlist{noitemsep, topsep=3pt}

\usepackage[raggedright]{titlesec}
\titleformat*{\section}{\large\sffamily\bfseries}
\titleformat*{\subsection}{\normalsize\sffamily\bfseries}
\titleformat{\paragraph}[hang]
  {\normalfont\normalsize\bfseries}{\theparagraph}{1em}{}
\titlespacing*{\paragraph}{0pt}{1.5ex plus 0.2ex}{0.5em}
\usepackage{authblk}

\usepackage[round,authoryear]{natbib}

\usepackage[
    colorlinks=true,
    linkcolor=linkblue,
    citecolor=linkblue,
    urlcolor=linkblue
]{hyperref}
\newtheorem{theorem}{Theorem}[section]
\newtheorem{proposition}[theorem]{Proposition}
\newtheorem{corollary}[theorem]{Corollary}
\newtheorem{lemma}[theorem]{Lemma}

\theoremstyle{definition}
\newtheorem{definition}[theorem]{Definition}
\newtheorem{example}[theorem]{Example}

\theoremstyle{remark}
\newtheorem{remark}[theorem]{Remark}

\newenvironment{myproof}
  {\par\smallskip\noindent\textit{Proof.}\;}
  {\hfill$\blacksquare$\par\medskip}

\newcommand{\K}{\mathcal{K}}
\newcommand{\I}{\mathcal{I}}
\newcommand{\Fringe}{\mathrm{Fringe}}
\newcommand{\down}[1]{\mathord{\downarrow}#1}
\DeclareMathOperator*{\argmax}{arg\,max}

\newcommand{\cmark}{\ensuremath{\checkmark}}
\newcommand{\xmark}{\ensuremath{\times}}

\title{\textbf{Exploration Space Theory}\\[0.3em]
\Large Formal Foundations for Prerequisite-Aware Location-Based Recommendation}

\author[1]{Madjid Sadallah}
\affil[1]{\small LIRIS, Universit\'{e} Claude Bernard Lyon~1, CNRS,
          43 boulevard du 11 novembre 1918, 69622 Villeurbanne, France\\
          \texttt{madjid.sadallah@liris.cnrs.fr}}

\date{}

\begin{document}
\maketitle
\hrule\medskip

\begin{abstract}
Location-based recommender systems have achieved considerable sophistication, yet
none provides a formal, lattice-theoretic representation of prerequisite dependencies
among points of interest---the semantic reality that meaningfully experiencing certain
locations presupposes contextual knowledge gained from others---nor the structural
guarantees that such a representation entails. We introduce Exploration Space Theory
(EST), a formal framework that transposes Knowledge Space Theory into location-based
recommendation. We prove that the valid user exploration states---the order ideals of
a surmise partial order on points of interest---form a finite distributive lattice and
a well-graded learning space; Birkhoff's representation theorem, combined with the
structural isomorphism between lattices of order ideals and concept lattices, connects
the exploration space canonically to Formal Concept Analysis. These structural results
yield four direct consequences: linear-time fringe computation, a validity certificate
guaranteeing that every fringe-guided recommendation is a structurally sound next
step, sub-path optimality for dynamic-programming path generation, and provably
existing structural explanations for every recommendation. Building on these
foundations, we specify the Exploration Space Recommender System (ESRS)---a memoized
dynamic program over the exploration lattice, a Bayesian state estimator with beam
approximation and EM parameter learning, an online feedback loop enforcing the
downward-closure invariant, an incremental surmise-relation inference pipeline, and
three cold-start strategies, the structural one being the only approach in the
literature to provide a formal validity guarantee conditional on the correctness of
the inferred surmise relation. All results are established through proof and
illustrated on a fully traced five-POI numerical example.
\end{abstract}

\noindent\textbf{Keywords:}
\small
Knowledge Space Theory — Formal Concept Analysis — Distributive Lattice Theory — Location-Based Recommender Systems — Sequential Recommendation — Exploration Paths — Urban Computing — Dynamic Programming

\medskip\hrule\bigskip

\section{Introduction}
\label{sec:intro}

Urban exploration is fundamentally a process of structured discovery. A visitor arriving in an unfamiliar city does not face a featureless scatter of locations but a partially ordered space of experiences in which certain visits naturally build upon others. Understanding the architectural evolution of a historic quarter provides the contextual grounding necessary to appreciate a specialist gallery housed within it; experiencing a city's central food market gives experiential meaning to a subsequent tour of its artisanal producers. This ordering is not a temporal sequence imposed by geography or logistics, but a partial order of cognitive and experiential prerequisites: a semantic structure that should be explicitly represented in any recommendation system that aspires to support genuine discovery rather than mere retrieval.

Current Location-Based Recommender Systems (LBRS) have made substantial advances in personalization, context-awareness, and sequential modeling \citep{ricci2011introduction, adomavicius2005toward, yuan2013time}. Collaborative filtering identifies users with similar behavioral patterns; deep sequential models---ranging from Markov chains to Recurrent Neural Networks and Transformer-based architectures \citep{rendle2010fpmc, hidasi2016session, kang2018self, yang2022getnext}---capture transition regularities from historical trajectories with impressive empirical performance. Yet across all these paradigms, a critical dimension remains absent: an explicit, formally grounded, and interpretable model of why certain sequences of locations are more meaningful than others. These systems learn correlational shadows of an underlying semantic structure they cannot represent. They cannot distinguish a visit sequence driven by geographic proximity from one driven by a genuine prerequisite dependency. Consequently, users in unfamiliar environments frequently receive fragmented lists of popular venues rather than coherent, structured journeys \citep{gavalas2014survey}.

To address this structural gap, we turn to Knowledge Space Theory \citep[KST;][]{doignon1985knowledge, doignon1999knowledge, falmagne2011learning}. Originally designed to formalize knowledge acquisition in educational domains, KST defines a knowledge space as a family of states closed under union, structured by surmise relations in which mastering one item implies having mastered its prerequisites. Over four decades, this theory has proven effective in driving adaptive assessment and intelligent tutoring systems \citep{corbett1994knowledge, desmarais2012real}. The central thesis of the present work is that the mathematical apparatus of KST translates naturally, rigorously, and productively to the domain of urban exploration.

By transposing KST to LBRS, the representational paradigm shifts fundamentally. Rather than modeling a user solely through latent preference embeddings, we propose to represent the user's cumulative exploration state: a structured object encoding what the user has meaningfully visited and thereby determining what they are now ready to discover. We prove that any surmise relation over a finite set of points of interest generates a collection of valid exploration states forming a finite distributive lattice. This algebraic structure provides immediate algorithmic guarantees: linear-time computation of the exploration fringe, guaranteed valid fringe-guided transitions, and sub-path optimal dynamic-programming-based path recommendation. These guarantees are not heuristic claims but consequences of proven mathematical propositions.

Translating these results into a functional architecture, we propose the Exploration Space Recommender System (ESRS). The ESRS integrates a unified interest score that synthesizes user preferences, location properties, collaborative signals, and structural state-relative constraints. It employs a BLIM-inspired probabilistic model to handle uncertainty in user state estimation, and offers three KST-grounded cold-start strategies, one of which---purely structural cold-start---provides a formal validity guarantee
unavailable in any existing approach (see \S\ref{sec:coldstart} and Limitation~L5).

The present work is a conceptual and theoretical contribution. We provide a comprehensive architectural specification, a complete algorithmic pipeline, and a fully traced numerical example. Empirical evaluation against state-of-the-art baselines is explicitly identified as the primary direction for future work; the objective here is to establish the mathematical foundations of EST, prove their correctness, and demonstrate how the framework addresses structural limitations of current approaches that are not resolvable by purely statistical means.
\section{Background and Related Work}
\label{sec:background}

\subsection{Location-Based Recommendation Systems}
\label{sec:bg_lbrs}

Location-Based Recommender Systems sit at the intersection of geographic information
science, personalization, and machine learning. Their core task is to produce, for a
given user and spatiotemporal context, a ranked list of Points of Interest (POIs)
likely to be relevant and engaging.

\textbf{Collaborative filtering and geographic influence.}
The seminal observation of \citet{ye2011exploiting} is that the probability of
visiting a POI decreases as a power-law function of distance from a user's home or
current position. This spatial structure has become a standard component of LBRS
architecture \citep{schafer2007collaborative}.

\textbf{Context integration.}
\citet{yuan2013time} demonstrated that visited POI types vary significantly by time
of day and day of week. More broadly, \citet{adomavicius2011context} established
the principle of context-aware recommendation: interactions should be modeled as
triples $(u, i, c)$ rather than pairs $(u, i)$. Content-based filtering leverages
POI attribute structure \citep{pazzani2007content, lops2011content}; hybrid
architectures address individual weaknesses of each paradigm \citep{burke2002hybrid}.

\textbf{Deep learning and graph-based approaches.}
\citet{he2017neural} introduced Neural Collaborative Filtering, replacing the inner
product of matrix factorization with a multi-layer perceptron. Graph Neural Networks
extended these capabilities by explicitly modeling relational structure
\citep{fan2019graph, wu2021graphsurvey}, propagating information along user-POI
interaction graphs, social graphs, and category hierarchies.

\textbf{Persistent structural gap.}
Despite this sophistication, all the above systems fundamentally model individual POI
relevance or pairwise transition patterns. They have no mechanism to ask whether the
user is semantically \emph{ready} for a given POI---whether the cognitive and
experiential prerequisites for meaningfully engaging with it have been met. EST
addresses this gap by modeling the user's cumulative exploration state as a
first-class object and encoding semantic readiness through the surmise relation.

\subsection{Sequence-Aware and Session-Based Recommendation}
\label{sec:bg_sequential}

\textbf{Markov-chain approaches.}
\citet{rendle2010fpmc} introduced FPMC, combining personalized user factors with
first-order Markov transition patterns for next-POI prediction. The first-order
Markov assumption limits ability to capture long-range dependencies.

\textbf{RNN-based approaches.}
\citet{hidasi2016session} introduced GRU4Rec, using Gated Recurrent Units to model
session-level recommendation. \citet{feng2017deepmove} proposed DeepMove, using an
attention-based recurrent architecture separately modeling long-term and short-term
mobility patterns.

\textbf{Transformer-based approaches.}
Self-attention \citep{kang2018self, sun2019bert4rec} has superseded RNNs as the
dominant architecture for sequential recommendation. SASRec achieves state-of-the-art
performance on multiple benchmarks. GETNext \citep{yang2022getnext} integrates a
graph-based global POI transition model with a local Transformer encoder. STAN
\citep{luo2021stan} applies spatio-temporal self-attention weighting past check-ins
by spatial proximity and temporal recency.

\textbf{The fundamental representational boundary.}
All sequential models learn to predict the next item by identifying statistical
regularities in observed sequences. High co-occurrence does not imply a prerequisite
relationship, and a prerequisite relationship may not manifest as high co-occurrence
if most users satisfy it via indirect paths, if the prerequisite is universally
satisfied (no variance), or if training data is sparse. Sequential models therefore
cannot reliably infer, represent, or reason about prerequisite dependencies even if
such dependencies structure the data. EST introduces an orthogonal modeling
dimension---the formal prerequisite structure---that sequential models cannot
represent by design.

\subsection{The Tourist Trip Design Problem}
\label{sec:bg_ttdp}

The Tourist Trip Design Problem (TTDP) \citep{vansteenwegen2011orienteering,
gavalas2014survey} asks for a subset of POIs and an ordering maximizing total score
without violating a time budget. It is a generalization of the NP-hard Orienteering
Problem. \citet{lim2018personalized} introduced a personalized variant in which scores
are functions of user interest profiles, but retains the core assumption that POIs
are independent: there are no inter-POI dependencies making the desirability of one
contingent on having visited another.

ESRS and the TTDP differ in three structural dimensions. First, ESRS interest scores
depend on the user's current exploration state, which changes after each visit.
Second, ESRS models prerequisite dependencies as the primary structural element.
Third, the TTDP does not model cumulative exploration knowledge. Despite these
differences, the two frameworks are complementary: the ESRS path optimization component can incorporate TTDP-style time budget and time-window constraints as side constraints on
the DP state space, as we formalize in the algorithmic section
(Remark~\ref{rem:ttdp_dp} in \S\ref{sec:dp}).

\subsection{Knowledge Space Theory}
\label{sec:bg_kst}

Knowledge Space Theory (KST) was developed by \citet{doignon1985knowledge} to model
the structure of human knowledge and the process of knowledge acquisition.

A \emph{knowledge structure} $(Q, \mathcal{K})$ consists of a finite domain $Q$ of
items and a family $\mathcal{K} \subseteq 2^Q$ of \emph{knowledge states}, satisfying
$\emptyset \in \mathcal{K}$ and $Q \in \mathcal{K}$. A knowledge structure is a
\emph{knowledge space} if $\mathcal{K}$ is closed under arbitrary unions; it is a
\emph{learning space} if additionally closed under non-empty intersections and
satisfying the well-graded property \citep{falmagne2011learning}.

The \emph{surmise relation} $\preceq$ on $Q$ is defined by $q \preceq q'$ if
mastering $q'$ implies having mastered $q$. The family of all downward-closed subsets
of $(Q, \preceq)$ is a knowledge space \citep{doignon1999knowledge} and in fact a
learning space \citep{falmagne2011learning}.

The \emph{fringe} of a knowledge state $K$ is the set of items $q \notin K$ such
that $K \cup \{q\}$ is a valid state. In adaptive assessment, the fringe identifies
items that are one step ahead of the current state and have maximal diagnostic value.
This operational role is inherited directly by the exploration fringe in ESRS.

The BLIM \citep{falmagne1988stochastic} provides probabilistic inference machinery
connecting observed responses to knowledge state estimates via Bayes' rule with
item-specific false-positive ($\beta_i$) and false-negative ($\eta_i$) error
parameters. KST has been deployed in the ALEKS adaptive learning platform
\citep{falmagne2011learning} and applied to structured path recommendation in
technology-enhanced learning \citep{sadallah2023learning}.

\subsection{Formal Concept Analysis and Distributive Lattice Theory}
\label{sec:bg_fca}

Formal Concept Analysis (FCA) \citep{wille1982restructuring, ganter1999formal}
derives formal concepts from a binary object-attribute relation and organizes them
into a concept lattice---a complete lattice capturing the hierarchical structure of
the conceptual space.

The relevance of FCA to EST is through Birkhoff's representation theorem
\citep{birkhoff1937rings, davey2002lattices}: every finite distributive lattice is
isomorphic to the lattice of order ideals of a unique finite partial order. Since
the exploration space $(\K(Q, \preceq), \subseteq)$ is a finite distributive lattice
(Proposition~\ref{prop:birkhoff}), it is canonically identified with the ideal lattice
of $(Q, \preceq)$, connecting the framework to FCA and enabling direct use of concept
lattice visualization tools and incremental construction algorithms.

\subsection{Constraint-Based Recommender Systems}
\label{sec:bg_cbrs}

Constraint-based recommender systems \citep{felfernig2006constraint,
felfernig2011knowledge} (CBRS) encode a formalized domain model---a set of logical constraints
over item attributes---and use constraint reasoning to identify items consistent with
user requirements and domain rules.

EST shares with constraint-based recommendation the commitment to an explicit domain
model and associated interpretability and correctness guarantees. However, the
frameworks differ fundamentally. Constraint-based systems model compatibility
constraints at a single decision point---a one-shot filtering problem with a static
user requirement set. EST models temporal prerequisite dependencies across multiple
time steps: accessibility depends on the user's cumulative exploration trajectory,
a dynamic object changing with each confirmed visit. Moreover, constraint-based
systems use attribute-based constraints (``the item must have property $P$''), while
EST uses order-based prerequisites that generate the lattice structure underlying all
algorithmic guarantees.

\subsection{Cold-Start in Recommender Systems}
\label{sec:bg_coldstart}

Cold-start arises when interaction data is absent for a user or item
\citep{schein2002cold}. Standard mitigation strategies include content-based
bootstrapping, demographic filtering, active learning, and stereotype-based
initialization. In the LBRS context, a new user in an unfamiliar city faces cold-start
in two senses simultaneously: no city-specific interaction history and no established
position in the exploration state space.

The KST framework offers a qualitatively different cold-start mechanism:
$\Fringe(\emptyset) = \{q \in Q : \nexists\, p \in Q \text{ with } p \prec q\}$
comprises POIs with no prerequisites, unconditionally accessible to any user. These
are the natural entry points into the Exploration Space, identified directly from the
surmise relation without any interaction data. This structural strategy is not an approximation but an exact
consequence of the mathematical definition. It provides a formal
validity guarantee conditional on the correctness of the surmise
relation $\preceq$, a type of structural guarantee that purely
statistical cold-start strategies do not provide.

\subsection{Summary and Positioning}
\label{sec:bg_summary}

To the best of our knowledge, no existing LBRS paradigm
simultaneously models (i) explicit prerequisite dependencies
between POIs, (ii) a formal user exploration state evolving through the exploration
space as visits occur, (iii) multi-step path recommendation grounded in both
preferences and structural accessibility, and (iv) a principled cold-start strategy
derived from the mathematical structure of the domain model. The combination of these
four properties is the defining contribution of EST.
Table~\ref{tab:bg_positioning} summarizes this positioning.

\begin{table}[htb]
\centering
\small
\caption{Structural positioning of ESRS relative to surveyed LBRS paradigms.
\cmark\ = supported by design; $\circ$ = partially or approximately supported;
\xmark\ = not supported. All entries reflect design properties, not performance claims.}
\label{tab:bg_positioning}
\setlength{\tabcolsep}{5pt}
\begin{tabular}{lccccc}
\toprule
\textbf{Structural dimension}
  & \makecell{\textbf{CF/}\\\textbf{Content}}
  & \makecell{\textbf{Markov/}\\\textbf{RNN/Transf.}}
  & \makecell{\textbf{TTDP}}
  & \makecell{\textbf{CBRS}}
  & \makecell{\textbf{ESRS}\\\textbf{(ours)}} \\
\midrule
Explicit prerequisite dependencies      & \xmark   & \xmark   & \xmark   & $\circ$  & \cmark \\
Dynamic structured exploration state    & \xmark   & \xmark   & \xmark   & \xmark   & \cmark \\
Formal mathematical foundation          & $\circ$  & $\circ$  & \cmark   & \cmark   & \cmark \\
Multi-step path respecting dependencies & \xmark   & $\circ$  & \cmark   & \xmark   & \cmark \\
Structural cold-start from domain model & \xmark   & \xmark   & \xmark   & $\circ$  & \cmark \\
\bottomrule
\end{tabular}
\end{table}

\section{Formal Foundations of Exploration Space Theory}
\label{sec:foundations}

\subsection{Core Definitions and Conceptual Translation}
\label{sec:core_defs}

\subsubsection{The Conceptual Parallel}

In KST, a \emph{knowledge state} is a set of items a learner has mastered. The
collection of all realizable states is not the full power set but a structured
family reflecting prerequisite dependencies. EST transposes this to the urban
exploration context: (i) an \emph{exploration item} is a Point of Interest; (ii)
an \emph{exploration state} is a set of POIs a user has \emph{meaningfully} visited;
(iii) some POIs presuppose others in the sense that the experiential value of a visit
is substantially reduced without prior engagement with their prerequisites.

\subsubsection{Primary Definitions}

\begin{definition}[Exploration Structure and Exploration Space]
\label{def:exploration_structure}
An \emph{exploration structure} is a pair $(Q, \K)$ where $Q$ is a finite non-empty
set of \emph{exploration items} (Points of Interest) and $\K \subseteq 2^Q$ is a
family of \emph{exploration states}, satisfying $\emptyset \in \K$ and $Q \in \K$.

An exploration structure $(Q, \K)$ is an \emph{exploration space} if $\K$ is closed
under arbitrary unions:
\[
  \forall\, \{K_\lambda\}_{\lambda \in \Lambda} \subseteq \K:\quad
  \bigcup_{\lambda \in \Lambda} K_\lambda \in \K.
\]
It is an \emph{exploration learning space} if $\K$ is additionally closed under
arbitrary non-empty intersections (we adopt the strong formulation of learning
spaces as in \citet[Ch.~3]{falmagne2011learning}):
\[
  \forall\, \emptyset \neq \{K_\lambda\}_{\lambda \in \Lambda} \subseteq \K:\quad
  \bigcap_{\lambda \in \Lambda} K_\lambda \in \K.
\]
\end{definition}

\begin{definition}[Surmise Relation on Exploration Items]
\label{def:surmise}
A \emph{surmise relation} $\preceq$ on $Q$ is a partial order: reflexive, transitive,
and antisymmetric. The semantic interpretation is: $q \preceq q'$ means that
engaging meaningfully with $q'$ presupposes prior meaningful engagement with $q$.
We write $q \prec q'$ for the strict part and $q \lessdot q'$ for the covering
relation. The directed graph of covering relations is the \emph{Hasse diagram}
$H = (Q, E_H)$.
\end{definition}

Reflexivity encodes that a POI always presupposes itself. Transitivity encodes the
chaining of prerequisites: if engaging with the Art Gallery requires the City Museum,
and the Rooftop Bar requires the Art Gallery, then the Rooftop Bar also requires the
City Museum. Antisymmetry rules out circular prerequisites. In the presence
of cycles, items would collapse into equivalence classes under
mutual dependence, and the intended prerequisite semantics
would be ill-defined.

\begin{definition}[Valid Exploration State, Order Ideal, and Principal Ideal]
\label{def:valid_state}
A subset $K \subseteq Q$ is a \emph{valid exploration state} with respect to
$(Q, \preceq)$ if it is \emph{downward closed} (an \emph{order ideal}):
\[
  \forall\, q, q' \in Q:\quad
  q \preceq q' \;\text{and}\; q' \in K \;\Longrightarrow\; q \in K.
\]
The \emph{principal ideal} generated by $q \in Q$ is:
\[
  \down{q} = \{p \in Q : p \preceq q\}.
\]
We denote by $\K(Q, \preceq)$ the family of all valid exploration states.
\end{definition}

\subsubsection{Types of Prerequisite Dependency}

The surmise relation may encode four qualitatively different types of dependency in
the urban context: (i) \emph{geographical containment} (a specialist archive nested
within an old town presupposes an embodied sense of place); (ii) \emph{thematic
progression} (an artisanal cooperative is most meaningful to a visitor who has
developed, at the city's main market, a sense of the broader food culture);
(iii) \emph{institutional or logistical sequencing} (a visitor center provides
orientation scaffolding for subsequent exploration); (iv) \emph{data-driven
inference} (sequential regularities in historical trajectory data serve as an
empirical signal, subject to expert validation).

\subsubsection{Structural Properties}

Two properties of the surmise relation are important. \textbf{Sparsity}: most pairs
of POIs have no prerequisite relationship. Sparsity directly controls the
computational tractability of the exploration space and the degree to which
recommendations are structurally constrained. \textbf{Semantic vs.\ Euclidean
topology}: the surmise relation is entirely independent of geographic distance.
Two adjacent POIs may be semantically unrelated under $\preceq$, while two distant
POIs may be strongly prerequisite-linked. The surmise relation encodes a topology
of meaning orthogonal to the physical city.

\subsection{The Exploration Space as a Well-Graded Learning Space}
\label{sec:learning_space}

\begin{proposition}[Exploration Space is a Well-Graded Learning Space]
\label{prop:knowledge_space}
Let $Q$ be a finite non-empty set and $\preceq$ a surmise relation on $Q$.
The pair $(Q, \K(Q,\preceq))$ is an exploration learning space satisfying:
\begin{enumerate}[label=(\roman*), noitemsep]
  \item $\emptyset \in \K$ and $Q \in \K$;
  \item $\K$ is closed under arbitrary unions;
  \item $\K$ is closed under arbitrary non-empty intersections;
  \item $\K$ is \emph{well-graded}: for any $K \in \K$ with $K \neq \emptyset$,
        there exists a sequence
        $\emptyset = K_0 \subset K_1 \subset \cdots \subset K_m = K$
        with $m = |K|$, where $K_{j+1} \setminus K_j = \{p_j\}$ and
        $p_j \in \Fringe(K_j)$ for all $j$.
\end{enumerate}
\end{proposition}

\begin{myproof}
\textbf{(i).} $\emptyset$ satisfies downward closure vacuously. $Q$ contains every
element of the domain by definition.

\textbf{(ii).} Let $K = \bigcup_{\lambda} K_\lambda$. If $q \preceq q'$ and $q' \in K$,
then $q' \in K_{\lambda_0}$ for some $\lambda_0$. Since $K_{\lambda_0}$ is an order
ideal, $q \in K_{\lambda_0} \subseteq K$.

\textbf{(iii).} Let $K = \bigcap_\lambda K_\lambda$ (non-empty family). If $q \preceq q'$
and $q' \in K$, then $q' \in K_\lambda$ for every $\lambda$, so $q \in K_\lambda$
for every $\lambda$, giving $q \in K$.

\textbf{(iv).} We use the following claim: for any non-empty order ideal $K$ and
any maximal element $m$ of $K$ under $\preceq$, the element $m$ belongs to
$\Fringe(K \setminus \{m\})$. Indeed, for any $p \prec m$ strictly, downward closure
of $K$ gives $p \in K$, and $p \neq m$ gives $p \in K \setminus \{m\}$.
Moreover $K \setminus \{m\}$ is itself an order ideal: for any
$q \preceq q'$ with $q' \in K \setminus \{m\}$, downward closure
of $K$ gives $q \in K$, and $q = m$ would require $m \preceq q'$,
i.e.\ $m \prec q'$ (since $q' \neq m$), with $q' \in K$---contradicting
the maximality of $m$ in $K$.
Hence $q \neq m$ and $q \in K \setminus \{m\}$.
Starting from $K$, iteratively remove maximal elements to produce a decreasing chain
terminating at $\emptyset$ in exactly $|K|$ steps; reversing gives the required
increasing sequence.
\end{myproof}

\begin{remark}[What the Learning Space Property Enables]
\label{rem:kst_full}
Part~(ii) establishes that $\K(Q,\preceq)$ is an exploration space, ensuring the
structure is non-trivial. Part~(iii), closure under intersection, makes $(\K,\subseteq)$
a complete lattice (Proposition~\ref{prop:birkhoff}), enabling the memoization
strategy in the DP algorithm: states are lattice elements identifiable via their
set representation without ambiguity. Part~(iv) is the guarantee most directly
relevant to recommendation experience: for any valid state $K$, the system can
always reach $K$ by single-step fringe additions, and the recommendation pipeline
can never reach a valid state from which no single-step progress is possible (unless
the user is at $Q$). An additional consequence is that $(\K,\subseteq)$ is
\emph{graded}: all maximal chains between two comparable states $K \subseteq K'$
have the same length $|K'| - |K|$. This follows from
Proposition~\ref{prop:knowledge_space}(iv): applying the well-graded property
to $K$ gives a chain of length $|K|$ from $\emptyset$ to $K$, and applied to
$K'$ gives a chain of length $|K'|$ from $\emptyset$ to $K'$; every maximal
chain from $K$ to $K'$ therefore has length $|K'|-|K|$, independent of the
path taken. The ``distance'' between exploration states is thus well-defined
and path-independent.
\end{remark}

\subsection{Distributive Lattice Structure and Connection to Formal Concept Analysis}
\label{sec:lattice}

\begin{proposition}[Distributive Lattice]
\label{prop:birkhoff}
The exploration space $(\K(Q,\preceq), \subseteq)$ is a \emph{finite distributive
lattice} with meet $K \wedge K' = K \cap K'$, join $K \vee K' = K \cup K'$,
bottom $\hat{0} = \emptyset$, and top $\hat{1} = Q$.
\end{proposition}

\begin{myproof}
Proposition~\ref{prop:knowledge_space}(ii,iii) establishes that
$K \cup K'$ and $K \cap K' \in \K$ for any $K, K' \in \K$.
$K \cup K'$ is the least upper bound: it contains both $K$ and
$K'$, and any $K'' \in \K$ satisfying $K \subseteq K''$ and
$K' \subseteq K''$ satisfies $K \cup K' \subseteq K''$.
Symmetrically, $K \cap K'$ is the greatest lower bound.
Distributivity follows from the set-theoretic identity
$K \cup (K' \cap K'') = (K \cup K') \cap (K \cup K'')$
applied to the set-theoretic realizations of the lattice
operations.
\end{myproof}

\begin{remark}[Birkhoff's Representation Theorem and Join-Irreducibles]
\label{rem:birkhoff}
Birkhoff's theorem \citep{birkhoff1937rings, davey2002lattices} states that every
finite distributive lattice is isomorphic to the lattice of order ideals of its
poset of join-irreducible elements.

\textbf{Identification of join-irreducibles.}
We claim the join-irreducible elements of $(\K(Q,\preceq), \subseteq)$ are precisely
the principal ideals $\down{q}$ for $q \in Q$.

\textit{$\down{q}$ is an order ideal}: if $p \preceq p'$ and $p' \in \down{q}$,
then $p' \preceq q$ and by transitivity $p \preceq q$, so $p \in \down{q}$.

\textit{$\down{q}$ is join-irreducible}: suppose $\down{q} = I_1 \cup I_2$ with
$I_1, I_2 \in \K$ and $I_1, I_2 \subsetneq \down{q}$ (non-trivial decomposition).
Since $q \in I_1 \cup I_2$, say $q \in I_1$. Then $I_1$ contains $q$ and is an
order ideal, so $\down{q} \subseteq I_1$, giving $I_1 = \down{q}$, contradicting
the hypothesis $I_1 \subsetneq \down{q}$.

Every join-irreducible is a principal ideal: in a lattice of order
ideals ordered by inclusion, a join-irreducible element $K$
must have a unique maximal element $m$. Indeed, if $K$ had two
distinct maximal elements $m_1, m_2$, then
$K = \down{m_1} \cup \down{m_2}$ would be a non-trivial join
decomposition, contradicting join-irreducibility. Hence $K = \down{m}$.

The map $\varphi: q \mapsto \down{q}$ is an order isomorphism from $(Q, \preceq)$
to the join-irreducible poset of $(\K, \subseteq)$: injectivity follows from
antisymmetry, and order preservation from transitivity.

Applying Birkhoff's theorem, $(\K(Q,\preceq), \subseteq)$ is the unique finite
distributive lattice (up to isomorphism) whose join-irreducible poset is $(Q, \preceq)$.

\textbf{Practical consequences.}
(i) \emph{Compactness}: the entire lattice is encoded by the $n$-element partial
order $(Q, \preceq)$, enabling fringe computation without state enumeration.
(ii) \emph{Connection to FCA}: the lattice $(\K, \subseteq)$ is formally identical
in structure to a concept lattice, making the full toolkit of FCA---visualization
software, incremental construction algorithms---directly applicable.
(iii) \emph{Hasse diagram as computational substrate}: covering relations in
$(Q, \preceq)$ suffice for all ESRS computations, as shown in
Proposition~\ref{prop:fringe_complexity}.
\end{remark}

\subsection{The Exploration Fringe}
\label{sec:fringe}

\begin{definition}[Exploration Fringe]
\label{def:fringe}
Given a valid exploration state $K \in \K(Q,\preceq)$, the \emph{exploration fringe} is:
\[
  \Fringe(K) \;=\; \bigl\{q \in Q \setminus K \;\big|\;
  \down{q} \setminus \{q\} \subseteq K\bigr\}.
\]
Equivalently, $\Fringe(K)$ is the set of minimal elements of
$(Q \setminus K, \preceq\!\restriction_{Q\setminus K})$.
\end{definition}

\textbf{Operational interpretation.}
$q \in \Fringe(K)$ if and only if $K \cup \{q\} \in \K$ (proved in
Lemma~\ref{lem:fringe_valid}). The fringe precisely identifies the set of POIs
the system can recommend as next steps without creating an invalid exploration state.

\begin{lemma}[Fringe Transitions Preserve Validity]
\label{lem:fringe_valid}
For any $K \in \K(Q,\preceq)$ and $q \in \Fringe(K)$, we have
$K \cup \{q\} \in \K(Q,\preceq)$.
\end{lemma}

\begin{myproof}
Let $K' = K \cup \{q\}$. Let $p \preceq q'$ with $q' \in K'$. If $q' \in K$, then
$p \in K \subseteq K'$ by downward closure of $K$. If $q' = q$: if $p = q$ then
$p \in K'$; if $p \neq q$ then $p \prec q$ strictly, so $p \in \down{q} \setminus
\{q\} \subseteq K \subseteq K'$ by the fringe condition.
\end{myproof}

\begin{proposition}[Efficient Fringe Computation]
\label{prop:fringe_complexity}
Let $H = (Q, E_H)$ be the Hasse diagram of $\preceq$ and assume $O(1)$ membership
testing for $K$. Given any valid state $K$, the fringe $\Fringe(K)$ can be computed
in $O(n + |E_H|)$ time. When $K$ is updated by adding a single item $q^*$, the
fringe can be updated incrementally in $O(\mathrm{deg}_H^+(q^*))$ time.
\end{proposition}

\begin{myproof}
\textbf{Batch computation.} The fringe condition $q \in \Fringe(K)$ requires all
lower covers of $q$ in $H$ to be in $K$ (checking lower covers suffices because in a finite poset every
predecessor of $q$ lies below some lower cover of $q$, and
downward closure then propagates membership through transitivity). 
Algorithm: (1) initialize
counter array $\mathrm{cnt}[q] \leftarrow 0$ for all $q \in Q \setminus K$
($O(n)$); (2) for each $(p, q) \in E_H$ with $p, q \notin K$, increment $\mathrm{cnt}[q]$
($O(|E_H|)$); (3) return $\{q \in Q \setminus K : \mathrm{cnt}[q] = 0\}$ ($O(n)$).

\textbf{Incremental update.} Assuming the counter array $\mathrm{cnt}$ is maintained across
state updates, when $q^*$ is added to $K$: for each successor $s$
with $q^* \lessdot s$, decrement $\mathrm{cnt}[s]$; if $\mathrm{cnt}[s]$ reaches
$0$ and $s \notin K \cup \{q^*\}$, add $s$ to the fringe. Remove $q^*$ from the
fringe. Cost: $O(\mathrm{deg}_H^+(q^*))$.
\end{myproof}

\begin{figure}[htb]
\centering
\begin{tikzpicture}[
  node distance=1.6cm and 2.0cm,
  item/.style={draw=ForestGreen!70!black, thick, rounded corners,
               minimum width=2.4cm, minimum height=0.6cm,
               font=\small, fill=ForestGreen!8},
  state/.style={draw=MidnightBlue, thick, rounded corners,
                minimum width=2.6cm, minimum height=0.55cm,
                font=\scriptsize, fill=MidnightBlue!5},
  arrow/.style={-Stealth, thick},
  sarrow/.style={-Stealth, MidnightBlue, thick}
]

\node[item] (q1) at (0,0)    {$q_1$: City Museum};
\node[item] (q2) at (3.5,0)  {$q_2$: Latin Quarter}; 
\node[item] (q4) at (0,1.7)  {$q_4$: Art Gallery};
\node[item] (q3) at (3.5,1.7){$q_3$: Med.\ Library};
\node[item] (q5) at (0,3.4)  {$q_5$: Rooftop Bar};

\draw[arrow] (q1)--(q4) node[midway,left,font=\scriptsize]{$\prec$};
\draw[arrow] (q4)--(q5) node[midway,left,font=\scriptsize]{$\prec$};
\draw[arrow] (q2)--(q3) node[midway,right,font=\scriptsize]{$\prec$};

\node[draw=none,fill=none,font=\small\itshape] at (1.75,-1)
  {(a) Hasse diagram of $(Q, \preceq)$};


\begin{scope}[xshift=8.5cm] 
\node[state] (bot)   at (1.5,0)   {$\emptyset$};
\node[state] (s1)    at (-0.2,1.5) {$\{q_1\}$};     
\node[state] (s2)    at (3.2,1.5)  {$\{q_2\}$};     

\node[state] (s14)   at (-1.8,3)  {$\{q_1,q_4\}$};   
\node[state] (s12)   at (1.5,3)   {$\{q_1,q_2\}$};  
\node[state] (s23)   at (4.8,3)   {$\{q_2,q_3\}$};  

\node[state] (sdots) at (1.5,4.5) {$\dots$ (12 states total)};
\node[state] (top)   at (1.5,6)   {$Q = \{q_1,q_2,q_3,q_4,q_5\}$};

\draw[sarrow](bot)--(s1.south);  \draw[sarrow](bot)--(s2.south);
\draw[sarrow](s1.north)--(s12.south);  \draw[sarrow](s1.north)--(s14.south);
\draw[sarrow](s2.north)--(s12.south);  \draw[sarrow](s2.north)--(s23.south);

\draw[sarrow](s12.north)--(sdots.south); 
\draw[sarrow](s14.north)--(sdots.south);
\draw[sarrow](s23.north)--(sdots.south);
\draw[sarrow](sdots.north)--(top.south);

\node[draw=none,fill=none,font=\small\itshape] at (1.5,-1)
  {(b) Lattice $(\mathcal{K}, \subseteq), |\mathcal{K}|=12$};
\end{scope}

\end{tikzpicture}
\caption{(a) Hasse diagram of the surmise relation on
$Q = \{q_1,\ldots,q_5\}$ with covering relations $q_1 \lessdot q_4 \lessdot q_5$
and $q_2 \lessdot q_3$. Items $q_1$ and $q_2$ have no predecessors and form
$\Fringe(\emptyset)$.
(b) The resulting distributive lattice of valid exploration states. By
Remark~\ref{rem:birkhoff}, the join-irreducible elements are the five principal
ideals $\down{q_1} = \{q_1\}$, $\down{q_2} = \{q_2\}$, $\down{q_3} = \{q_2,q_3\}$,
$\down{q_4} = \{q_1,q_4\}$, $\down{q_5} = \{q_1,q_4,q_5\}$. The 12 valid states result from multiplicativity of ideal counts:
since the poset is the disjoint union of two independent chains,
its ideal lattice is the direct product of their ideal lattices.
The chain $q_1 \prec q_4 \prec q_5$
generates 4 ideals and $q_2 \prec q_3$ generates 3, giving $4 \times 3 = 12$.}
\label{fig:hasse}
\end{figure}
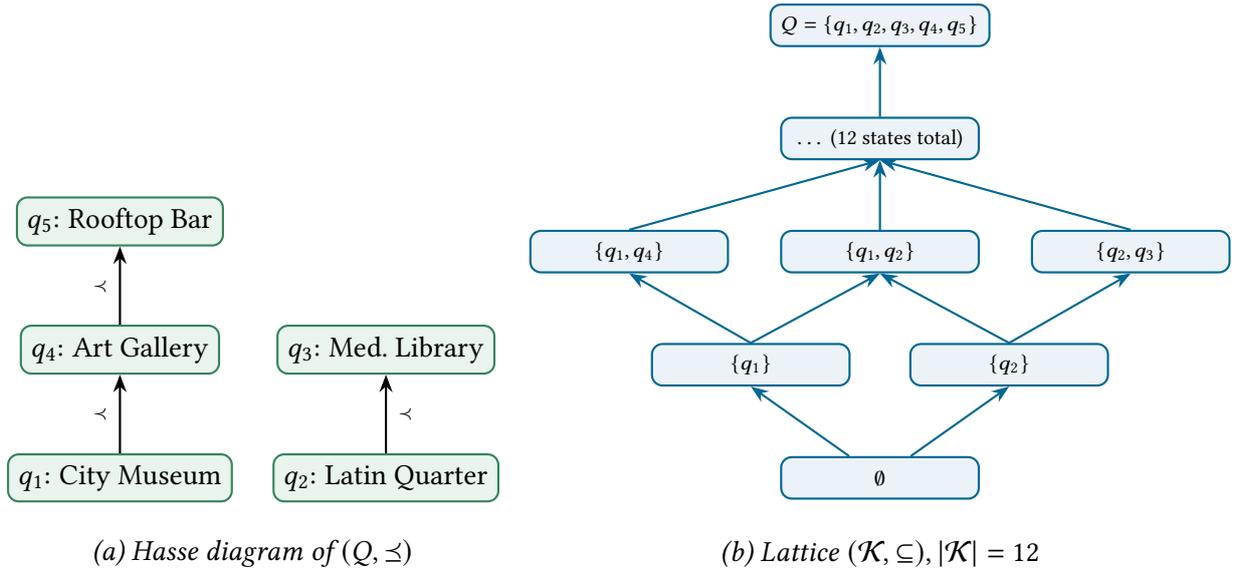

\subsection{State Transitions and Path Validity}
\label{sec:path_validity}

\begin{corollary}[Path Validity]
\label{cor:path_valid}
Let $K_0 \in \K(Q,\preceq)$ and let $\pi = (q_1, \ldots, q_k)$ satisfy, for each $j$:
$q_j \in \Fringe(K_{j-1})$ and $K_j = K_{j-1} \cup \{q_j\}$.
Then $K_j \in \K(Q,\preceq)$ for all $j$, and the items $q_1, \ldots, q_k$ are
pairwise distinct.
\end{corollary}

\begin{myproof}
Validity: by induction, Lemma~\ref{lem:fringe_valid} gives $K_j \in \K$ from
$K_{j-1} \in \K$ and $q_j \in \Fringe(K_{j-1})$. Distinctness: $q_j \in Q \setminus
K_{j-1}$ implies $q_j \notin \{q_1, \ldots, q_{j-1}\} \subseteq K_{j-1}$.
\end{myproof}

Corollary~\ref{cor:path_valid} is the safety certificate of the ESRS recommendation
algorithm. Any path generated by sequential fringe selection (i) respects the
prerequisite structure by construction, without any explicit constraint-checking step,
and (ii) never recommends the same POI twice.

A further structural consequence---sub-path optimality of the dynamic program over
$(\K, \subseteq)$---is established in Proposition~\ref{prop:subpath} once the
Bellman formulation is introduced (\S\ref{sec:dp}).

\section{System Architecture}
\label{sec:arch}

The Exploration Space Recommender System (ESRS) comprises three interconnected
components (Figure~\ref{fig:architecture}): the \textbf{Location Database}, the
\textbf{User Modeling Component}, and the \textbf{Recommendation Engine}. These
interact through the Exploration Space $(Q, \K, \preceq)$.

\begin{figure}[htb]
\centering
\begin{tikzpicture}[
  >=Stealth,
  node distance=1.4cm and 0.8cm,
  every node/.style={align=center, font=\small},
  box/.style={
    draw=MidnightBlue!70, fill=MidnightBlue!3, thick, rounded corners=2pt,
    minimum width=4.5cm, minimum height=1.3cm, text width=4.3cm
  },
  cbox/.style={
    draw=ForestGreen!70, fill=ForestGreen!3, thick, rounded corners=2pt,
    minimum width=4.5cm, minimum height=1.3cm, text width=4.3cm
  },
  arrow/.style={->, thick, MidnightBlue!80},
  darrow/.style={->, thick, ForestGreen!80},
  labelstyle/.style={font=\footnotesize, text=black!80}
]

\node[cbox] (es) {\textbf{Exploration Space $(Q, \mathcal{K}, \preceq)$} \\ \footnotesize Fringe, lattice, transitions};
\node[box]  (ra) [above=of es] {\textbf{Recommendation Engine} \\ \footnotesize Assessment $\to$ Fringe $\to$ DP / Rank};
\node[cbox] (ui) [above=of ra] {\textbf{User Interface \& Feedback} \\ \footnotesize Paths, rankings, check-ins};

\node[box]  (db) [below left=1.2cm and -0.8cm of es] {\textbf{Location Database} \\ \footnotesize POIs, attributes, surmise $\preceq$};
\node[box]  (um) [below right=1.2cm and -0.8cm of es] {\textbf{User Modeling} \\ \footnotesize $E_u(t)$, $P_u(K)$, $\mathbf{p}_u(t)$};


\draw[arrow] (db.north) -- ([xshift=-1.2cm]es.south) 
      node[labelstyle, left, midway, xshift=-0.1cm] {\textit{POI data}};
\draw[arrow] (um.north) -- ([xshift=1.2cm]es.south) 
      node[labelstyle, right, midway, xshift=0.1cm] {\textit{user state}};

\draw[arrow] (es) -- (ra) node[labelstyle, right, midway] {\textit{structure}};
\draw[arrow] (ra) -- (ui) node[labelstyle, right, midway] {\textit{recommendations}};

\draw[arrow] (db.west) -- ++(-0.6,0) |- (ra.west) 
      node[labelstyle, pos=0.25, sloped, above] {\textit{attributes}};

\draw[darrow] (ui.east) .. controls +(2.8,-1.5) and +(2.8,1.5) .. (um.east)
      node[labelstyle, midway, right, xshift=0.2cm] {\textit{feedback loop}};

\end{tikzpicture}
\caption{High-level architecture of ESRS. The Exploration Space $(Q, \K, \preceq)$ serves as the shared structural backbone, populated from the Location Database and navigated by the Recommendation Engine in light of the User Modeling Component.} \label{fig:architecture}
\end{figure}
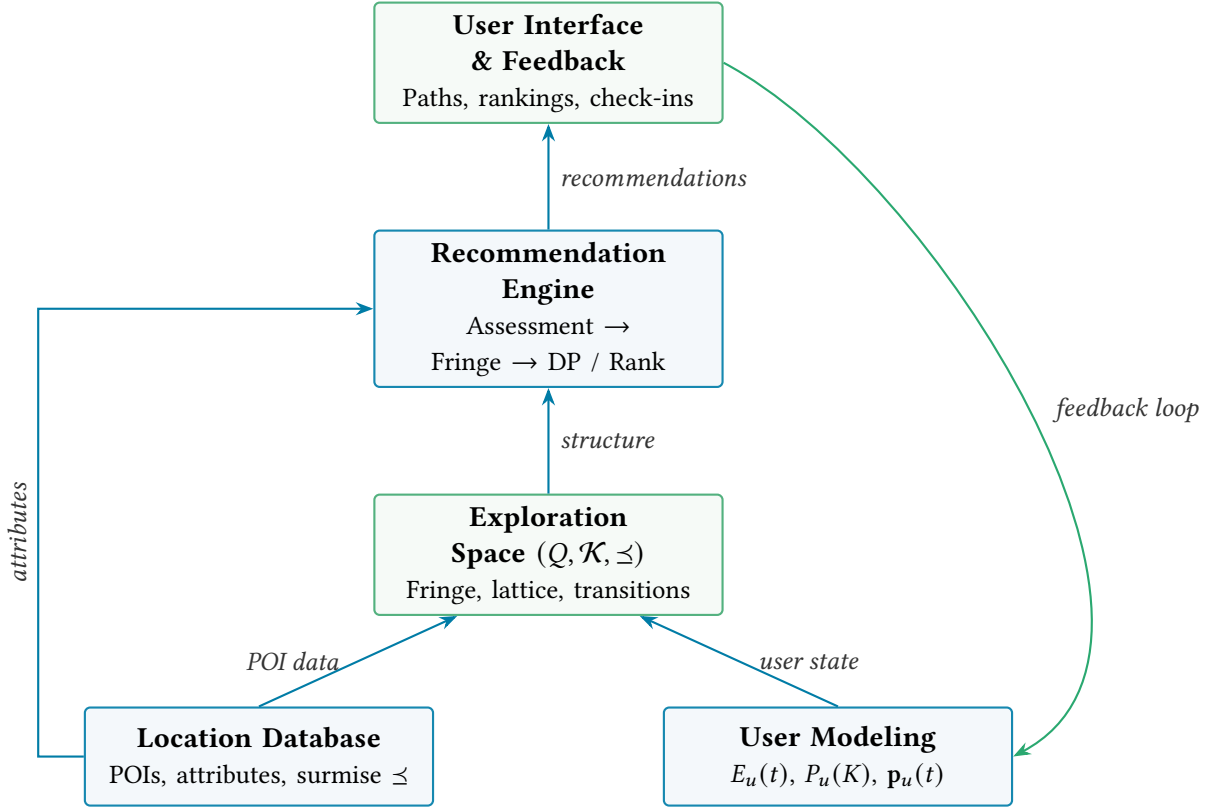

\subsection{Location Database}
\label{sec:locdb}

The surmise relation $\preceq$ is treated as a first-class citizen of the data model,
stored explicitly as directed edges alongside standard POI attributes. Each POI
$L_i \in Q$ is characterized by: unique identifier and coordinates $(x_i, y_i)$;
hierarchical category $\mathrm{Cat}_i$; normalized popularity $\mathrm{Pop}_i \in [0,1]$
and review score $\mathrm{Rev}_i \in [0,1]$; static context (opening hours, admission
policy); dynamic context (event feeds); and surmise links $(L_i \prec L_j)$ derived
from Algorithm~\ref{alg:prefixspan} subject to expert validation.

For new POIs, we define a soft neighborhood function:
\[
N_f(L_i) = \{L_j \in Q : \mathrm{sim}_{\mathrm{cat}}(L_i, L_j) > \theta_1 \;\text{or}\; \mathrm{dist}(L_i, L_j) < \theta_2\},
\]
where $\mathrm{sim}_{\mathrm{cat}}$ measures category overlap (e.g., Jaccard on category
hierarchy) and $\mathrm{dist}$ is geographic distance. This neighborhood is used for
initial placement and for identifying candidate POIs for potential surmise relations.

A graph database backend (e.g., Neo4j) is appropriate for efficient DAG traversal,
integrated with a geospatial layer (e.g., PostGIS) for spatial queries.

\subsection{User Modeling Component}
\label{sec:usermodel}

The User Modeling Component maintains three inter-related objects.

\paragraph*{Confirmed exploration state $E_u(t)$}
The set of POIs user $u$ has been confirmed to have meaningfully visited by time $t$,
operationalized via high-confidence engagement signals (deliberate check-in, dwell
time above threshold $\theta_d^+$, explicit rating). A strict invariant is maintained:
$E_u(t)$ is updated \emph{exclusively} upon reception of a high-confidence signal
and never overwritten by a probabilistic estimate. This preserves the structural
integrity of the exploration space; since only fringe items are ever recommended
(Algorithm~\ref{alg:esrs}) and Corollary~\ref{cor:path_valid} ensures fringe-guided
transitions remain within $\K$, the confirmed state can never violate the surmise
relation.

\paragraph*{Probabilistic state distribution $P_u$ and working state $\hat{E}_u$}
Because users visit POIs without registering check-ins, ESRS maintains a probability
distribution $P_u : \K \to [0,1]$ inspired by BLIM \citep{falmagne1988stochastic}.
Upon observing an interaction pattern $\mathbf{r}$, the distribution is updated by
Bayes' rule:
\[
  P_u(K \mid \mathbf{r}) \;\propto\; P(\mathbf{r} \mid K) \cdot P_u(K),
\]
where the BLIM likelihood factorizes as:
\[
  P(\mathbf{r} \mid K) = \prod_{i:\, r_i = 1}
    \bigl[(1 - \beta_i)^{\mathbb{1}[i \in K]}\,
           \beta_i^{\mathbb{1}[i \notin K]}\bigr]
    \cdot
    \prod_{i:\, r_i = 0}
    \bigl[\eta_i^{\mathbb{1}[i \in K]}\,
          (1 - \eta_i)^{\mathbb{1}[i \notin K]}\bigr],
\]
with false-positive rate $\beta_i$ (the probability of a positive signal without
a meaningful visit) and false-negative rate $\eta_i$ (the probability of no digital
signal despite a meaningful visit). The MAP estimate
$\hat{E}_u = \argmax_{K \in \K} P_u(K \mid \mathbf{r})$ is the \emph{working state}:
used for fringe computation and path optimization, discarded at end of each cycle.

When $|\K|$ is too large for exact enumeration, \emph{beam-search estimation}
retains only the $B$ most probable states and renormalizes. For $c$ independent
components with chain heights $h_1,\ldots,h_c$, we have $|\K| = \prod_{i=1}^c (h_i+1)$;
for example, with $c=5$ components each of height $h_i=3$,
$|\K| \leq 4^5 = 1024$, which is fully tractable. Setting $B \in \{100, 500\}$
achieves near-optimal MAP estimation in practice \citep{desmarais2012real}.

\paragraph*{Preference vector $\mathbf{p}_u(t)$}
Individual affinities for POIs, updated via exponential moving average:
\[
  p_{u,i}(t+1) = p_{u,i}(t) + \ell_r \cdot
  \bigl(I_{\mathrm{obs}} - p_{u,i}(t)\bigr), \qquad \ell_r \in (0,1],
\]
where $I_{\mathrm{obs}} \in [0,1]$ encodes normalized dwell time or rating.
For POIs not yet encountered, a latent factor model trained on the interaction matrix
$R \approx U \cdot V^T$ provides the prior.

\subsection{The Unified Interest Score}
\label{sec:interest}

\begin{equation}
\label{eq:interest}
\I(u, L_i, \hat{E}_u) \;=\;
w_{\alpha} \cdot \mathrm{Pref}_{u,i}
+ w_{\beta} \cdot \mathrm{Prop}_i
+ w_{\gamma} \cdot \mathrm{Collab}_{u,i}
+ w_{\delta} \cdot \mathrm{Rel}(L_i, \hat{E}_u),
\end{equation}
with $w_{\alpha}+w_{\beta}+w_{\gamma}+w_{\delta}=1$, all coefficients non-negative, and all
component scores in $[0,1]$.

\paragraph*{User preference $\mathrm{Pref}_{u,i}$}
A weighted combination of explicit feedback (ratings, bookmarks) and implicit
behavioral proxies (dwell time, click-through, check-in frequency), with matrix
factorization estimate $\hat{R}_{ui}$ as prior for unvisited POIs.

\paragraph*{Location properties $\mathrm{Prop}_i = w_c C_i + w_r \mathrm{Pop}_i + w_s \mathrm{Rev}_i$}
where $w_c + w_r + w_s = 1$ with all coefficients non-negative, ensuring
$\mathrm{Prop}_i \in [0,1]$.
Category relevance $C_i \in [0,1]$ measures the overlap between $\mathrm{Cat}_i$
and the user's inferred category preferences (e.g., Jaccard similarity between
$\mathrm{Cat}_i$ and $\{\mathrm{Cat}_j : L_j \in \hat{E}_u\}$, with a
content-based prior for new users).
Normalized popularity $\mathrm{Pop}_i$ and review score $\mathrm{Rev}_i$
are defined in \S\ref{sec:locdb}.
This component is particularly important in cold-start regimes where no
preference data is available.

\paragraph*{Collaborative signal $\mathrm{Collab}_{u,i}$}
User similarity measured via regularized Jaccard:
\[
\mathrm{Sim}_{u,v} = \frac{|E_u \cap E_v| + \varepsilon}{|E_u \cup E_v| + \varepsilon},
\]
where $\varepsilon > 0$ (typically $\varepsilon = 1$) ensures stability; for
$E_u = E_v = \emptyset$, $\mathrm{Sim}_{u,v} = 1$. For new users, a hybrid similarity
combines this with preference-vector cosine:
\[
\mathrm{Sim}^{\mathrm{hybrid}}_{u,v} = \lambda_s \mathrm{Sim}_{u,v} + (1-\lambda_s)
\frac{\mathbf{p}_u \cdot \mathbf{p}_v}{\|\mathbf{p}_u\|\|\mathbf{p}_v\|},
\]
where $\lambda_s : \mathbb{N} \to [0,1]$ satisfies $\lambda_s(0) = 0$ and
$\lambda_s(t) \to 1$ as $|E_u(t)|$ grows
(e.g., $\lambda_s(t) = 1 - e^{-|E_u(t)|/\kappa}$ for $\kappa > 0$),
ensuring $\mathrm{Sim}^{\mathrm{hybrid}}_{u,v}$ reduces to cosine similarity
at system bootstrap. The collaborative signal is then:
\[
\mathrm{Collab}_{u,i} = \frac{\sum_{v \in \mathcal{U}_u} \mathrm{Sim}^{\mathrm{hybrid}}_{u,v} \cdot \mathrm{Pref}_{v,i}}{\sum_{v \in \mathcal{U}_u} \mathrm{Sim}^{\mathrm{hybrid}}_{u,v}}.
\]

\begin{proposition}[Cold-Start Robustness of the Collaborative Signal]
\label{prop:collab_coldstart}
For any user $u$ with $E_u(t_0) = \emptyset$ and
$\mathbf{p}_u(t_0) \neq \mathbf{0}$, the collaborative signal
$\mathrm{Collab}_{u,i}$ is well-defined. It is strictly positive
for any POI $L_i$ with at least one neighbor
$v^* \in \mathcal{U}_u$ satisfying $\mathrm{Pref}_{v^*,i} > 0$
and $\mathbf{p}_{v^*} \neq \mathbf{0}$.
\end{proposition}

\begin{myproof}
With $E_u(t_0) = \emptyset$, we have $|E_u(t_0)| = 0$, so by
definition $\lambda_s(0) = 0$ and
$\mathrm{Sim}^{\mathrm{hybrid}}_{u,v} =
  \frac{\mathbf{p}_u \cdot \mathbf{p}_v}
       {\|\mathbf{p}_u\|\,\|\mathbf{p}_v\|}$,
the cosine similarity.  Since preference scores lie in $[0,1]$,
the cosine is non-negative; it is strictly positive whenever
$\mathbf{p}_u \neq \mathbf{0}$ and $\mathbf{p}_v \neq \mathbf{0}$.
The denominator
$\sum_{v \in \mathcal{U}_u} \mathrm{Sim}^{\mathrm{hybrid}}_{u,v}$
is strictly positive whenever $|\mathcal{U}_u| \geq 1$ and
$\mathbf{p}_u \neq \mathbf{0}$, so $\mathrm{Collab}_{u,i}$
is well-defined.
\textit{Non-zero:} by hypothesis there exists $v^*\in\mathcal{U}_u$
with $\mathrm{Pref}_{v^*,i}>0$.  If $\mathbf{p}_{v^*}\neq\mathbf{0}$
and $\mathbf{p}_u\neq\mathbf{0}$, then
$\mathrm{Sim}^{\mathrm{hybrid}}_{u,v^*}>0$, and the numerator
$\sum_v \mathrm{Sim}^{\mathrm{hybrid}}_{u,v}\cdot\mathrm{Pref}_{v,i}
\geq \mathrm{Sim}^{\mathrm{hybrid}}_{u,v^*}\cdot
\mathrm{Pref}_{v^*,i} > 0$.
\end{myproof}

\paragraph*{State-relative structural accessibility $\mathrm{Rel}(L_i, \hat{E}_u)$}
\begin{equation}
\label{eq:rel}
\mathrm{Rel}(L_i, \hat{E}_u) \;=\;
\frac{|\down{L_i} \cap \hat{E}_u|}{|\down{L_i}|}.
\end{equation}
This equals 1 if and only if $L_i \in \hat{E}_u$: since $\hat{E}_u$ is an order
ideal, $L_i \in \hat{E}_u$ implies $\down{L_i} \subseteq \hat{E}_u$, giving
$|\down{L_i} \cap \hat{E}_u| = |\down{L_i}|$; conversely, since $L_i \in \down{L_i}$ by reflexivity of $\preceq$, $L_i \notin \hat{E}_u$ implies $L_i \notin \down{L_i} \cap \hat{E}_u$,
so $|\down{L_i} \cap \hat{E}_u| \leq |\down{L_i}| - 1 < |\down{L_i}|$.
For a fringe item $L_i \in \Fringe(\hat{E}_u)$, we have
$\down{L_i} \setminus \{L_i\} \subseteq \hat{E}_u$ and $L_i \notin \hat{E}_u$,
so $\mathrm{Rel}(L_i, \hat{E}_u) = (|\down{L_i}|-1)/|\down{L_i}|$: the score is
strictly less than 1 but approaches 1 as the depth of $L_i$ grows. For
prerequisite-free fringe items ($|\down{L_i}|=1$, i.e., $\down{L_i}=\{L_i\}$),
the score is 0 regardless of $\hat{E}_u$---they are structurally accessible but
have no visited prerequisites to reference. The score provides a graded measure
of proximity to accessibility for non-fringe items; for fringe items it encodes
prerequisite depth, with deeper items receiving a higher score.

The Recommendation Engine (Algorithm~\ref{alg:esrs}) restricts
ranking and path optimization to $F = \Fringe(\hat{E}_u)$.
Within the fringe, $\mathrm{Rel}$ contributes to
$\mathcal{I}$ via the $w_{\delta}$ term: it assigns higher
scores to items whose prerequisite chains are deeper, providing
a gradient of structural proximity even among accessible items.
For prerequisite-free fringe items ($|\down{L_i}|=1$), this
term is uniformly zero; the other three components of
$\mathcal{I}$ carry the full discriminative weight.

\subsection{Cold-Start Strategies}
\label{sec:coldstart}

\subsubsection{New Users}

\textbf{Strategy~1---Structural onboarding.}
$\Fringe(\emptyset) = \{q \in Q : \nexists\, p \in Q \text{ with } p \prec q\}$
comprises POIs with no prerequisites, unconditionally accessible to any user.

\begin{proposition}[Structural Cold-Start Guarantee]
\label{prop:coldstart_guarantee}
For any surmise relation $(Q, \preceq)$ and any user $u$ with $E_u(t_0) = \emptyset$,
the set $\Fringe(\emptyset)$ is non-empty, and $\{q\} \in \K(Q, \preceq)$ for all
$q \in \Fringe(\emptyset)$.
\end{proposition}

\begin{myproof}
Since $Q$ is finite and non-empty, any partial order on $Q$ has at least one
minimal element under $\prec$ (the finiteness of $Q$ guarantees well-foundedness,
which precludes infinite descending chains and hence ensures minimal elements exist),
so $\Fringe(\emptyset)$ is non-empty. For $q \in \Fringe(\emptyset)$:
$\down{q} \setminus \{q\} = \emptyset \subseteq \emptyset$, so $\{q\}$ is valid
by Lemma~\ref{lem:fringe_valid}.
\end{myproof}

A lightweight questionnaire maps stated interests to categories in $\Fringe(\emptyset)$,
initializing recommendations using the $w_{\beta}$ component of the interest score. This constitutes knowledge-based recommendation \citep{burke2002hybrid}: principled recommendations before any
user data is collected, with a formal validity guarantee conditional on the correctness of the inferred surmise relation~$\preceq$ (see Limitation~L5).

\textbf{Strategy~2---Popularity bootstrapping.}
When the user declines the questionnaire, set $w_{\alpha}=w_{\gamma}=w_{\delta}=0$, $w_{\beta}=1$,
ranking $\Fringe(\emptyset)$ by inherent attractiveness. Weights are progressively
rebalanced toward preference- and collaborative-based components as interactions
accumulate, via $w_{\alpha}(t) = w_{\alpha}^\infty(1 - e^{-t/\tau})$.

\textbf{Strategy~3---Stereotype initialization.}
Map new users to exploration archetypes (\textit{Cultural Discoverer},
\textit{Culinary Explorer}, \textit{Historical Researcher}) based on available
contextual signals, initializing $\mathbf{p}_u(t_0)$ from the corresponding cluster
centroid \citep{schein2002cold}.

\subsubsection{New POIs}

New POI $L_{\mathrm{new}}$ is initially placed in $\Fringe(\emptyset)$ with no
surmise links (conservative placement guaranteeing no artificial structural constraints).
It is included in the soft neighborhood $N_f(L_j)$ for its most similar existing POIs
via the neighborhood function defined in \S\ref{sec:locdb}. As trajectory data
accumulates, Algorithm~\ref{alg:prefixspan} is re-executed with $L_{\mathrm{new}}$
included, and candidate surmise relations are integrated via the incremental update
of Proposition~\ref{prop:fringe_complexity} after expert validation.

\section{Recommendation Algorithm}
\label{sec:algorithm}

\paragraph*{Notational discipline.}
$E_u(t) \in \K$ denotes the \emph{confirmed exploration state}, updated exclusively
upon high-confidence signals. $\hat{E}_u \in \K$ denotes the \emph{working state}
(MAP estimate, computed at the start of each cycle and discarded at its end).
$P_u : \K \to [0,1]$ denotes the \emph{state distribution}. The confirmed state
is never directly overwritten by any probabilistic estimate.
We use $\eta_i \in (0,1)$ exclusively for the BLIM false-negative rate of item $L_i$,
and $\ell_r \in (0,1]$ for the preference learning rate.

\subsection{The Full ESRS Pipeline}
\label{sec:pipeline}

\begin{algorithm}[htb]
\caption{ESRS Recommendation Pipeline}
\label{alg:esrs}
\begin{algorithmic}[1]
\Require User $u$; confirmed state $E_u(t) \in \K$; distribution $P_u$;
         preference vector $\mathbf{p}_u(t)$;
         exploration space $(Q, \K, \preceq)$; Location Database $\mathcal{D}$;
         mode $\in \{\textsc{Path}, \textsc{Rank}\}$; max steps $k_{\max}$;
         beam width $B$; diversity weights $(w_N, w_D)$
\Ensure Recommended path $\pi^*$ or diversified ranked list $\mathcal{R}$

\State \textbf{// Phase 1: State Assessment}
\State $(P_u,\, \hat{E}_u) \leftarrow
  \Call{BLIMAssess}{P_u, \mathbf{r}_{\mathrm{recent}}, B}$
\hfill\Comment{$P_u$ updated in-place; $\hat{E}_u$ is the MAP estimate}
\State $\hat{E}_u \leftarrow \hat{E}_u \cup E_u(t)$
\hfill\Comment{Confirmed visits always present in working state}
\State Compute $\I(u, L_i, \hat{E}_u)$ via Eq.~\eqref{eq:interest} for all $L_i \in Q$
\hfill\Comment{Used directly in \textsc{Rank} mode; \textsc{Path} mode
recomputes $\I(u,q,K)$ at each DP node with the evolving state $K$
(Algorithm~\ref{alg:dp})}

\State \textbf{// Phase 2: Fringe Construction}
\State $F \leftarrow \Fringe(\hat{E}_u)$
\hfill\Comment{$O(n+|E_H|)$; by Definition~\ref{def:fringe}, $\Fringe(\hat{E}_u) \subseteq Q \setminus \hat{E}_u$; Proposition~\ref{prop:fringe_complexity}}
\If{$F = \emptyset$}
  \State \textbf{emit} ``Exploration complete; propose new semantic cluster or city area''
  \State \Return $\emptyset$
\EndIf

\State \textbf{// Phase 3a: Optimal Path}
\If{mode $= \textsc{Path}$}
  \State \Return $\Call{DPPath}{\hat{E}_u, k_{\max}, B}$
  \hfill\Comment{Algorithm~\ref{alg:dp}; fringe computed internally}
\EndIf

\State \textbf{// Phase 3b: Diversified Ranking}
\If{mode $= \textsc{Rank}$}
  \State \Return $\Call{DiverseRank}{u, F, \hat{E}_u, k_{\max}, w_N, w_D}$
  \hfill\Comment{Algorithm~\ref{alg:rank}}
\EndIf
\end{algorithmic}
\end{algorithm}

\begin{remark}[Confirmed state invariant]
\label{rem:invariant}
Line~3 merges $\hat{E}_u$ with $E_u(t)$, ensuring $E_u(t) \subseteq \hat{E}_u$.
Since $E_u(t) \in \K$ and $\hat{E}_u \in \K$, and $\K$ is closed under arbitrary
unions (Proposition~\ref{prop:knowledge_space}(ii)), the union $\hat{E}_u \cup E_u(t)
\in \K$. The working state is discarded at the end of the
pipeline and never written back to $E_u(t)$; only Algorithm~\ref{alg:feedback} may
update the confirmed state, exclusively on high-confidence signals.
\end{remark}

\subsection{Probabilistic State Estimation}
\label{sec:blim_algo}

\begin{definition}[Engagement signal and response vector]
\label{def:feedback_signals}
Let $A \subseteq Q$ be a set of assessed items. For each $L_i \in A$, the
\emph{engagement signal} $r_i \in \{0,1\}$ indicates positive engagement
($r_i = 1$: dwell time above $\theta_d$, explicit check-in, or user-initiated rating).
A signal $r_i = 1$ is \emph{high-confidence} if it additionally satisfies dwell
time above $\theta_d^+ > \theta_d$ or an explicit rating.
\end{definition}

The BLIM likelihood under local independence is:
\begin{equation}
\label{eq:blim_likelihood}
\ell(\mathbf{r} \mid K) = \prod_{i:\, r_i = 1}
  \bigl[(1 - \beta_i)^{\mathbb{1}[L_i \in K]} \cdot \beta_i^{\mathbb{1}[L_i \notin K]}\bigr]
  \cdot
  \prod_{i:\, r_i = 0}
  \bigl[\eta_i^{\mathbb{1}[L_i \in K]} \cdot (1-\eta_i)^{\mathbb{1}[L_i \notin K]}\bigr].
\end{equation}
The Bayesian posterior is $P_u(K \mid \mathbf{r}) \propto \ell(\mathbf{r} \mid K) \cdot P_u(K)$.

\begin{algorithm}[htb]
\caption{BLIM State Estimation (Exact and Beam Approximation)}
\label{alg:blim}
\begin{algorithmic}[1]
\Require Distribution $P_u$; response vector $\mathbf{r}$ over $A$;
  BLIM parameters $\{\beta_i, \eta_i\}$; beam width $B$
\Ensure Updated distribution $P_u'$; MAP estimate $\hat{E}_u$
\State $\mathcal{B} \leftarrow \text{top-}B \text{ states by } P_u(\cdot)$
\ForAll{$K \in \mathcal{B}$}
  \State $\tilde{P}(K) \leftarrow P_u(K)$
  \State $\tilde{P}(K) \leftarrow \tilde{P}(K) \cdot
    \prod_{i \in A:\, r_i=1}
      \bigl[(1-\beta_i)^{\mathbb{1}[L_i \in K]} \beta_i^{\mathbb{1}[L_i \notin K]}\bigr]$
  \State $\tilde{P}(K) \leftarrow \tilde{P}(K) \cdot
    \prod_{i \in A:\, r_i=0}
      \bigl[\eta_i^{\mathbb{1}[L_i \in K]} (1-\eta_i)^{\mathbb{1}[L_i \notin K]}\bigr]$
\EndFor
\State $Z \leftarrow \sum_{K \in \mathcal{B}} \tilde{P}(K)$;
       $P_u'(K) \leftarrow \tilde{P}(K)/Z$ for $K \in \mathcal{B}$, $0$ otherwise
\State \Return $P_u'$,\ $\hat{E}_u \leftarrow \argmax_{K \in \mathcal{B}} P_u'(K)$
\end{algorithmic}
\end{algorithm}

\begin{remark}[Validity of the beam approximation]
\label{rem:beam_validity}
With $B < |\K|$, the approximation error is controlled by the prior mass outside
the beam, $\varepsilon_B = 1 - \sum_{K \in \mathcal{B}} P_u(K)$. If the true MAP
state $K^*$ is in $\mathcal{B}$---which holds whenever the prior is sufficiently
concentrated---the MAP estimate is exact. The MAP estimate $\hat{E}_u$ is always a
member of $\K$ (the argmax is taken within $\mathcal{B} \subseteq \K$).
When the EM algorithm is used for parameter estimation, the beam approximation discards
the probability mass of states outside $\mathcal{B}$. Provided $B$ is chosen large enough
to capture the vast majority of the posterior mass (e.g., $B \ge 100$), the impact on
the M-step estimates is negligible; this is the common practice in BLIM applications
\citep{desmarais2012real}.
\end{remark}

\subsubsection{EM Parameter Estimation}
\label{sec:blim_em}

When ground-truth labels are unavailable, EM alternates between:
\begin{itemize}[noitemsep]
  \item \textbf{E-step}: compute $P_u(K \mid \mathbf{r}^{(s)})$ for each training
        sequence using current parameters.
  \item \textbf{M-step}: update by closed-form maximization:
  \[
    \hat{\beta}_i = \frac{\sum_s \sum_{K \not\ni L_i} P_u(K \mid \mathbf{r}^{(s)})
      \cdot \mathbb{1}[r_i^{(s)}=1]}{\sum_s \sum_{K \not\ni L_i} P_u(K \mid \mathbf{r}^{(s)})},
  \]
  with a symmetric expression for $\hat{\eta}_i$, and
  $P_u(K) \propto \sum_s P_u(K \mid \mathbf{r}^{(s)})$.
\end{itemize}
Convergence to a local maximum is guaranteed by the standard EM monotonicity argument
under exact inference \citep{falmagne1988stochastic}; when the beam approximation is
used in the E-step, the discarded posterior mass breaks the monotonicity guarantee,
and convergence is not formally assured (see Remark~\ref{rem:beam_validity}).

\subsection{Dynamic Programming for Path Recommendation}
\label{sec:dp}

\subsubsection{Bellman Formulation and Value Function}

The value function $V : \K \times \{0,\ldots,k_{\max}\} \to \mathbb{R}$ satisfies
the following Bellman recursion \citep{bellman1957dynamic}:
\begin{equation}
\label{eq:bellman}
V(K, j) = \begin{cases}
  0 & \text{if } j = 0 \text{ or } \Fringe(K) = \emptyset, \\
  \displaystyle\max_{q \in \Fringe(K)}\bigl[\I(u, q, K) + V(K \cup \{q\}, j-1)\bigr]
  & \text{otherwise.}
\end{cases}
\end{equation}
The policy is $\pi^*(K, j) = \argmax_{q \in \Fringe(K)}\bigl[\I(u, q, K) + V(K \cup \{q\}, j-1)\bigr]$.

\begin{proposition}[Sub-Path Optimality]
\label{prop:subpath}
Let $\pi^* = (q_1, \ldots, q_k)$ be a path achieving $V(\hat{E}_u, k)$ under
Bellman recursion~\eqref{eq:bellman}, with $K_0 = \hat{E}_u$ and
$K_j = K_{j-1} \cup \{q_j\}$. For any $0 \leq j < k$, the sub-path
$(q_{j+1}, \ldots, q_k)$ achieves $V(K_j, k-j)$, i.e., it is optimal for the
sub-problem starting from state $K_j$ with remaining horizon $k - j$.
\end{proposition}

\begin{myproof}
Suppose for contradiction that for some $j$ there exists a path $(q'_{j+1}, \ldots, q'_k)$
from $K_j$ with \\ $\sum_{\ell=j+1}^k \I(u, q'_\ell, K'_{\ell-1}) > \sum_{\ell=j+1}^k
\I(u, q_\ell, K_{\ell-1})$, where $K'_\ell = K_j \cup \{q'_{j+1},\ldots,q'_\ell\}$.
Then the path \\ $(q_1, \ldots, q_j, q'_{j+1}, \ldots, q'_k)$ from $\hat{E}_u$ achieves
value strictly greater than $V(\hat{E}_u, k)$. By Corollary~\ref{cor:path_valid} all
intermediate states $K'_\ell$ are valid, so this path is feasible---contradicting
the optimality of $\pi^*$.
\end{myproof}

\subsubsection{Memoization over the Lattice}

The value $V(K, j)$ depends on $K$ only through the identity of the set $K$, not on
the path used to reach $K$. Consequently, if two traversal sequences reach the same
state $K$ with the same remaining horizon $j$, $V(K, j)$ is computed only once.
States are identified as sorted tuples of POI identifiers (canonical representation
admitting $O(n)$ hashing). A top-down recursive call on $\textsc{DPVal}(K,j)$ triggers
$\textsc{DPVal}(K\cup\{q\},j{-}1)$ before returning, guaranteeing
that every sub-problem with a strictly larger state set (or smaller
remaining horizon) is resolved before its parent.
Each pair $(K,j)$ is evaluated at most once via the memo table.

\begin{algorithm}[htb]
\caption{ESRS Path Recommendation via Memoized DP}
\label{alg:dp}
\begin{algorithmic}[1]
\Require Working state $\hat{E}_u$; horizon $k_{\max}$; beam width $B$
\Ensure  Optimal path $\pi^*$; optimal value $V^*$

\State Initialize $\mathtt{memo}\leftarrow\{\}$,
       $\mathtt{pred}\leftarrow\{\}$

\Function{DPVal}{$K,j$}
  \If{$(K,j)\in\mathtt{memo}$} \Return $\mathtt{memo}[(K,j)]$
  \EndIf
  \If{$j=0$ \textbf{or} $\Fringe(K)=\emptyset$}
    \State \Return $0$
  \EndIf
  \State $C\leftarrow\Call{TopBFringe}{K,B}$
  \State $q^*\leftarrow\argmax_{q\in C}
      \bigl[\I(u,q,K)+\Call{DPVal}{K\cup\{q\},j{-}1}\bigr]$
    \State $\mathtt{pred}[(K,j)]\leftarrow q^*$
    \State $\mathtt{memo}[(K,j)]\leftarrow
      \I(u,q^*,K) + \mathtt{memo}[(K\cup\{q^*\},j{-}1)]$
    \State \Return $\mathtt{memo}[(K,j)]$
\EndFunction

\State $V^*\leftarrow\Call{DPVal}{\hat{E}_u,k_{\max}}$
\State $\pi^*\leftarrow[\,]$; $K_c\leftarrow\hat{E}_u$;
       $j_c\leftarrow k_{\max}$
\While{$j_c>0$ \textbf{and} $(K_c,j_c)\in\mathtt{pred}$}
  \State $q^*\leftarrow\mathtt{pred}[(K_c,j_c)]$;
         $\pi^*\mathtt{.append}(q^*)$;
         $K_c\leftarrow K_c\cup\{q^*\}$; $j_c\leftarrow j_c-1$
\EndWhile
\State \Return $\pi^*$, $V^*$

\Procedure{TopBFringe}{$K,B$}
  \State $F_K\leftarrow\Fringe(K)$
  \If{$B=\infty$ \textbf{or} $|F_K|\leq B$} \Return $F_K$
  \Else\ \Return top-$B$ items in $F_K$ by $\I(u,\cdot,K)$
  \EndIf
\EndProcedure
\end{algorithmic}
\end{algorithm}

\begin{proposition}[Correctness of Memoized DP]
\label{prop:dp_correct}
Algorithm~\ref{alg:dp} with $B = \infty$ computes $V(\hat{E}_u, k_{\max})$ correctly
and returns a path achieving this value. Each state-horizon pair $(K,j)$ is evaluated
at most once.
\end{proposition}

\begin{myproof}
\textbf{Correctness}: by structural induction on $j$.
\textit{Base case $j = 0$}: \textsc{DPVal}$(K, 0)$ returns $0$,
matching Eq.~\eqref{eq:bellman}.
\textit{Base case $\Fringe(K) = \emptyset$, $j > 0$}: \textsc{DPVal}$(K, j)$
returns $0$, matching the second clause of the base case in
Eq.~\eqref{eq:bellman}.
\textit{Inductive step}: assuming $\Fringe(K) \neq \emptyset$ and
\textsc{DPVal}$(K', j{-}1)$ returns $V(K', j{-}1)$ correctly for all $K'$,
the expression
\[
  v^* = \max_{q \in C}\bigl[\I(u,q,K) +
  \Call{DPVal}{K \cup \{q\}, j{-}1}\bigr]
\]
evaluates the Bellman recursion~\eqref{eq:bellman} correctly
(with $C = \Fringe(K)$ when $B = \infty$).

\textbf{Each pair evaluated at most once}: the memo table check
at the top of \textsc{DPVal} returns immediately on any revisit,
so each $(K,j)$ pair triggers the recursive computation at most
once. Termination follows because every recursive call increases
$|K|$ or decreases $j$, and both quantities are bounded
($|K| \leq |Q|$, $j \geq 0$).

\textbf{Backtracking}: the predecessor map $\mathtt{pred}$
stores, for each $(K,j)$, the $\argmax$ item; following it from
$(\hat{E}_u, k_{\max})$ recovers a path realizing
$V(\hat{E}_u, k_{\max})$.
Corollary~\ref{cor:path_valid} ensures all intermediate states
$K \cup \{q^*\}$ are valid.
\end{myproof}

\begin{remark}[Effect of Beam Width on Complexity and Optimality]
\label{rem:beam_dp}
With $B = \infty$:
$T_{\mathrm{exact}} = O\bigl(k_{\max} \cdot |\K_{\mathrm{reach}}(k_{\max})| \cdot
(n + |E_H|)\bigr)$.
With $B < \infty$:
$T_{\mathrm{beam}} = O\bigl(k_{\max} \cdot B \cdot (n + |E_H|)\bigr)$.
The beam search is guaranteed optimal when the optimal first step is among the top-$B$
fringe items by interest score, which is empirically robust when $w_{\delta} > 0$.
Proposition~\ref{prop:subpath} continues to hold under exact search; under beam
search, sub-path optimality holds conditionally on the beam containing the optimal
intermediate states.
\end{remark}

\begin{remark}[TTDP Integration via Augmented State Space]
\label{rem:ttdp_dp}
As noted in Section~\ref{sec:bg_ttdp}, TTDP-style temporal constraints (time budget
$T_{\max}$, time windows $[o_i, c_i]$) can be incorporated by augmenting the DP
state from $K$ to $(K, t_{\mathrm{elapsed}})$:
\[
  V(K, t, j) = \max_{\substack{q \in \Fringe(K)\\ t + \mathrm{dur}(q) + \mathrm{travel}(K,q) \leq T_{\max}}}
  \!\!\!\bigl[\I(u,q,K) + V(K \cup \{q\},\, t + \mathrm{dur}(q) + \mathrm{travel}(K,q),\, j-1)\bigr].
\]
The structural guarantees of Lemma~\ref{lem:fringe_valid} and Corollary~\ref{cor:path_valid}
are unaffected; the temporal constraint simply reduces the feasible fringe at each step.
\end{remark}

\subsection{Diversified Top-$k$ Ranking}
\label{sec:rankmode}

\begin{definition}[Intra-list Diversity]
\label{def:diversity}
For $\mathcal{R} \subseteq F$ already selected and $L_i \in F \setminus \mathcal{R}$:
\[
  D(L_i, \mathcal{R}) = 1 - \frac{1}{|\mathcal{R}|}\sum_{L_j \in \mathcal{R}}
    \mathrm{sim}_{\mathrm{cat}}(L_i, L_j) \cdot \exp(-\lambda \cdot \mathrm{dist}(L_i, L_j)),
\]
where $\lambda > 0$ is a scale parameter controlling the decay of geographic
proximity relative to categorical similarity (larger $\lambda$ penalizes
geographically distant pairs more strongly; $\lambda$ is set by cross-validation
or domain expertise).
Set $D(L_i, \emptyset) = 1$.
$\mathrm{sim}_{\mathrm{cat}}(L_i, L_j)$ and $\mathrm{dist}(L_i, L_j)$ for
$L_i, L_j \in F$ are precomputed in $O(|F|^2)$ before
Algorithm~\ref{alg:rank} is called, so each evaluation of $D$ costs $O(|\mathcal{R}|)$.
\end{definition}

The diversified score is:
\begin{equation}
\label{eq:rank_score}
S(u, L_i, \hat{E}_u, \mathcal{R}) =
  w_{\I} \cdot \I(u, L_i, \hat{E}_u)
  + w_N \cdot (1 - \mathrm{Pop}_i)
  + w_D \cdot D(L_i, \mathcal{R}),
\quad w_{\I} + w_N + w_D = 1.
\end{equation}

\begin{algorithm}[htb]
\caption{Diversified Top-$k$ Ranking (MMR-style)}
\label{alg:rank}
\begin{algorithmic}[1]
\Require User $u$; fringe $F$; working state $\hat{E}_u$; $k_{\max}$;
         weights $w_{\I}, w_N, w_D$
\Ensure Ranked list $\mathcal{R}$ of size $\min(k_{\max}, |F|)$
\State $\mathcal{R} \leftarrow [\,]$; $\mathcal{F}_r \leftarrow F$
\While{$|\mathcal{R}| < k_{\max}$ \textbf{and} $\mathcal{F}_r \neq \emptyset$}
  \State $L^* \leftarrow \argmax_{L_i \in \mathcal{F}_r} S(u, L_i, \hat{E}_u, \mathcal{R})$
  \State $\mathcal{R}\mathtt{.append}(L^*)$; $\mathcal{F}_r \leftarrow \mathcal{F}_r \setminus \{L^*\}$
\EndWhile
\State \Return $\mathcal{R}$
\end{algorithmic}
\end{algorithm}

\begin{definition}[Structural Serendipity]
\label{def:serendipity}
Let $\mathcal{C}_u = \{\mathrm{Cat}_i : L_i \in \hat{E}_u\}$. A POI
$L_i \in \Fringe(\hat{E}_u)$ is \emph{structurally serendipitous} if
$\mathrm{Cat}_i \notin \mathcal{C}_u$ and $|\down{L_i}| \geq 2$.
\end{definition}

\begin{remark}[Sufficient condition for structural serendipity]
\label{rem:serendipity}
A sufficient condition for $\Fringe(\hat{E}_u)$ to contain at least one
structurally serendipitous item is that there exists a covering relation
$L_j \lessdot L_i$ in $(Q,\preceq)$ such that $\mathrm{Cat}_j \in \mathcal{C}_u$,
$\mathrm{Cat}_i \notin \mathcal{C}_u$, and
$\down{L_i} \setminus \{L_i\} \subseteq \hat{E}_u$.
Under these conditions, $L_i \in \Fringe(\hat{E}_u)$ by
Definition~\ref{def:fringe}, $|\down{L_i}| \geq 2$ (since $L_j \prec L_i$
is a strict predecessor), and $\mathrm{Cat}_i \notin \mathcal{C}_u$, so $L_i$
is structurally serendipitous. This configuration arises naturally as
$\hat{E}_u$ grows: each fringe-guided transition into a new category exposes
cross-category covering relations as candidates.
\end{remark}

\subsection{Surmise Relation Inference}
\label{sec:seqpattern}

\begin{algorithm}[htb]
\caption{Surmise Relation Inference via Sequential Pattern Mining}
\label{alg:prefixspan}
\begin{algorithmic}[1]
\Require Trajectory database $S = \{\tau^{(1)},\ldots,\tau^{(m)}\}$;
         minimum support $\sigma$; confidence threshold $\tau_c$;
         significance level $\alpha_{\mathrm{stat}}$
\Ensure Partial order $\preceq$ on $Q$

\State $P_2 \leftarrow \{(a,b) : \mathrm{support}(\langle a,b\rangle, S) \geq \sigma\}$
\hfill\Comment{Via PrefixSpan \citep{pei2001prefixspan}}
\ForAll{$(a,b) \in P_2$}
  \State $n_a    \leftarrow |\{\tau \in S : a \in \tau\}|$
         \hfill\Comment{trajectories containing $a$}
  \State $n_{ab} \leftarrow |\{\tau \in S : a \text{ precedes } b \text{ in } \tau\}|$
         \hfill\Comment{trajectories containing $a$ before $b$}
\EndFor

\State $\mathrm{Cand} \leftarrow \emptyset$
\ForAll{$(a,b) \in P_2$}
  \State $\hat{c}_{ab} \leftarrow n_{ab} / n_a$;
         $p_{\mathrm{val}} \leftarrow \Call{BinomTest}{n_{ab}, n_a, \tau_c}$
  \If{$\hat{c}_{ab} \geq \tau_c$ \textbf{and} $p_{\mathrm{val}} \leq \alpha_{\mathrm{stat}}$}
    \State $\mathrm{Cand} \leftarrow \mathrm{Cand} \cup \{(a,b)\}$
  \EndIf
\EndFor

\State \textbf{// Cycle resolution BEFORE transitive closure}
\ForAll{$(a,b) \in \mathrm{Cand}$ with $(b,a) \in \mathrm{Cand}$}
  \If{$\hat{c}_{ab} > \hat{c}_{ba}$}
    \State Remove $(b,a)$ from $\mathrm{Cand}$
  \ElsIf{$\hat{c}_{ba} > \hat{c}_{ab}$}
    \State Remove $(a,b)$ from $\mathrm{Cand}$
  \Else\ \Comment{Exact tie: remove both to avoid arbitrary choice}
    \State Remove $(a,b)$ and $(b,a)$ from $\mathrm{Cand}$
  \EndIf
\EndFor
\State $\mathrm{SCCs} \leftarrow \Call{TarjanSCC}{\mathrm{Cand}}$
\ForAll{$C \in \mathrm{SCCs}$ with $|C| > 1$}
  \State Retain only the highest-confidence edge within $C$; remove all others
\EndFor
\hfill\Comment{$\mathrm{Cand}$ is now a DAG}

\State $\preceq \leftarrow \Call{FloydWarshall}{\mathrm{Cand}}$
\hfill\Comment{$O(n^3)$; DAG ensures partial order}
\State Flag pairs with $\hat{c}_{ab} < \tau_c^{\mathrm{high}}$ for human review;
       remove unvalidated pairs
\State \Return $\preceq$
\end{algorithmic}
\end{algorithm}

The one-sided binomial test in Step~2 guards against spurious high-confidence
estimates from rare items: it tests whether $n_{ab}$ is significantly above
$c_{ab} = \tau_c$, combining an empirical threshold with a statistical significance
requirement to avoid false positives from POI pairs with few co-occurrences.

\subsubsection{Incremental Surmise Update}

\begin{algorithm}[htb]
\caption{Incremental Surmise Relation Update}
\label{alg:incr_surmise}
\begin{algorithmic}[1]
\Require Current $\preceq$; Hasse diagram $H$; new trajectories $\Delta S$;
         thresholds; fringe counter array $\mathtt{cnt}$
\Ensure Updated $\preceq'$, $H'$, $\mathtt{cnt}'$
\State $\Delta P_2 \leftarrow$ new high-confidence pairs from $\Delta S$
       (Algorithm~\ref{alg:prefixspan}, Steps 1--5)
\ForAll{$(a,b) \in \Delta P_2$ not yet in $\preceq$}
  \If{$(b,a) \in \preceq$} Resolve conflict; skip if $a \prec b$ via transitivity
  \Else
    \State Add covering edge $(a \lessdot b)$ to $E_H$ if appropriate;
           update $\preceq$ transitively ($O(n^2)$ worst case: all predecessors
           of $a$ must be related to all successors of $b$ in the updated closure)
    \State Update fringe counters via Proposition~\ref{prop:fringe_complexity}
           incremental procedure ($O(\mathrm{deg}^+(b))$ per active user)
  \EndIf
\EndFor
\State Validate acyclicity; \Return $\preceq'$, $H'$, $\mathtt{cnt}'$
\end{algorithmic}
\end{algorithm}

\subsection{Online Feedback Processing and State Update}
\label{sec:feedback}

\begin{algorithm}[htb]
\caption{Online Feedback Processing and User Model Update}
\label{alg:feedback}
\begin{algorithmic}[1]
\Require Confirmed state $E_u(t)$; distribution $P_u$; preference $\mathbf{p}_u(t)$;
         event $(L_i, r_i, \mathtt{hc}_i)$; BLIM parameters; $\ell_r$; $H$; $\mathtt{cnt}$
\Ensure Updated $E_u(t')$, $P_u'$, $\mathbf{p}_u(t')$

\State $I_{\mathrm{obs}} \leftarrow r_i$ \Comment{Or normalized dwell/rating if available}
\State $p_{u,i}(t') \leftarrow p_{u,i}(t) + \ell_r \cdot (I_{\mathrm{obs}} - p_{u,i}(t))$
\hfill\Comment{Step 1: update preference}
\State $(P_u', \hat{E}_u) \leftarrow \Call{BLIMAssess}{P_u, (r_i)_{L_i}, B}$
\hfill\Comment{Step 2: update distribution}
\If{$r_i = 1$ \textbf{and} $\mathtt{hc}_i = \mathtt{true}$}
  \If{$L_i \in \Fringe(E_u(t))$ \textbf{or} $L_i \in E_u(t)$}
    \State $E_u(t') \leftarrow E_u(t) \cup \{L_i\}$
    \hfill\Comment{Valid by Lemma~\ref{lem:fringe_valid}}
    \ForAll{$s$: $L_i \lessdot s$ in $H$}
      $\mathtt{cnt}[s] \leftarrow \mathtt{cnt}[s] - 1$
    \EndFor
  \Else
    \State \textbf{log} warning; $E_u(t') \leftarrow E_u(t)$
    \hfill\Comment{Structural invariant preserved}
  \EndIf
\Else\ $E_u(t') \leftarrow E_u(t)$
\EndIf
\State \Return $E_u(t')$, $P_u'$, $\mathbf{p}_u(t')$
\end{algorithmic}
\end{algorithm}

The guard at Step~3 is the operational enforcement of the downward-closure invariant.
When a user engages with a non-fringe POI independently of the system, admitting
it to the confirmed state could violate downward closure if prerequisites are absent.
The guard rejects the update to $E_u(t)$; the distribution $P_u$ is still updated,
which may surface missing prerequisites as candidates in the next cycle. The
incremental fringe update costs $O(\mathrm{deg}_H^+(L_i))$.

\subsection{Consolidated Complexity Analysis}
\label{sec:complexity_summary}

\begin{table}[htb]
\centering
\small
\caption{Time complexity of ESRS algorithmic components. $n = |Q|$; $|E_H|$ = edges
in Hasse diagram; $|\K_{\mathrm{reach}}|$ = states reachable within $k_{\max}$ steps;
$B$ = beam width; $|A|$ = assessed items; $m$ = training trajectories;
$\bar{d}$ = average trajectory length.}
\label{tab:complexity}
\setlength{\tabcolsep}{6pt}
\begin{tabular}{lllp{4.5cm}}
\toprule
\textbf{Component} & \textbf{Algorithm} & \textbf{Time complexity} & \textbf{Notes} \\
\midrule
Batch fringe        & Prop.~\ref{prop:fringe_complexity} & $O(n + |E_H|)$ & Per state \\
Incremental fringe  & Prop.~\ref{prop:fringe_complexity} & $O(\mathrm{deg}_H^+(q^*))$ & Per transition \\[4pt]
BLIM exact          & Alg.~\ref{alg:blim} ($B=|\K|$) & $O(|\K|\cdot|A|)$ & Exponential worst case \\
BLIM beam           & Alg.~\ref{alg:blim} ($B<|\K|$) & $O(B\cdot|A|)$ & Approximate MAP \\[4pt]
DP exact            & Alg.~\ref{alg:dp} ($B=\infty$) & $O(k_{\max}\cdot|\K_{\mathrm{reach}}|\cdot(n+|E_H|))$ & Optimal path \\
DP beam             & Alg.~\ref{alg:dp} ($B<\infty$) & $O(k_{\max}\cdot B\cdot(n+|E_H|))$ & Approx.\ optimal \\[4pt]
MMR ranking         & Alg.~\ref{alg:rank} & $O(k_{\max}^2\cdot|F|)$ & $|F|=|\Fringe(\hat{E}_u)|$; pairwise $\mathrm{sim}_{\mathrm{cat}}$ and $\mathrm{dist}$ precomputed in $O(|F|^2)$ \\[4pt]
Surmise (batch)     & Alg.~\ref{alg:prefixspan} & $O(m\bar{d} + n^3)$ & PrefixSpan + Floyd-Warshall \\
Surmise (incr.)     & Alg.~\ref{alg:incr_surmise} & $O(n^2)$ & Per new pair (transitive closure update) \\[4pt]
Feedback            & Alg.~\ref{alg:feedback} & $O(B\cdot|A| + \mathrm{deg}_H^+(L_i))$ & Per event \\
Matrix factorization & \S\ref{sec:usermodel} & $O(|\mathcal{D}_{\mathrm{obs}}|\cdot k)$ per epoch & Offline init. \\
\bottomrule
\end{tabular}
\end{table}

\begin{remark}[Decomposability of $|\K|$]
\label{rem:decomposability}
If $(Q,\preceq)$ decomposes into $c$ connected components
$(Q_1,\preceq_1),\ldots,(Q_c,\preceq_c)$ (all elements of distinct components
being incomparable under $\preceq$), then
$|\K(Q,\preceq)| = \prod_{i=1}^c |\K(Q_i,\preceq_i)|$.
\textit{Proof}: any order ideal $I$ of the disjoint union decomposes uniquely
as $(I \cap Q_1, \ldots, I \cap Q_c)$, each $I \cap Q_k$ being an order ideal
of $(Q_k,\preceq_k)$; the correspondence is a bijection, giving the product
formula. For a chain of height $h_i$ ($h_i + 1$ elements), $|\K(Q_i)| = h_i+1$,
so for $c$ independent chains of heights $h_1,\ldots,h_c$:
$|\K(Q,\preceq)| = \prod_{i=1}^c (h_i+1)$.
\end{remark}

The scalability bottleneck is the exponential dependence of exact BLIM inference
and exact DP on $|\K|$ and $|\K_{\mathrm{reach}}|$.
For $c$ independent components
with chain heights $h_1,\ldots,h_c$: $|\K| = \prod_{i=1}^c (h_i+1)$. For $c=5$,
$h_i=3$: $|\K| \leq 4^5 = 1024$ (fully tractable); for $c=15$, $h_i=4$:
$|\K| \leq 5^{15} \approx 3\times10^{10}$ (requiring beam approximation). The
incremental update procedures operate in near-linear time per event, ensuring
responsiveness at interaction frequency regardless of exploration space size.

\section{Worked Example}
\label{sec:example}

This section traces user $u$ through a complete execution of the ESRS pipeline on
the five-POI instance of Figure~\ref{fig:hasse}. The objectives are: (1) verify that
interest score computations are well-defined; (2) confirm that each intermediate
exploration state is a valid order ideal, verifying Corollary~\ref{cor:path_valid};
(3) trace the memoization mechanism of Algorithm~\ref{alg:dp}; (4) verify the
sub-path optimality of Proposition~\ref{prop:subpath}; and (5) illustrate the
feedback loop of Algorithm~\ref{alg:feedback}.

\subsection{Instance Setup}
\label{sec:example_setup}

$Q = \{q_1,\ldots,q_5\}$: $q_1$ = City Museum, $q_2$ = Latin Quarter, $q_3$ =
Medieval Library, $q_4$ = Art Gallery, $q_5$ = Rooftop Bar. Surmise relation:
$q_1 \prec q_4 \prec q_5$ and $q_2 \prec q_3$, all other pairs incomparable.
$|\K| = 4 \times 3 = 12$.

Principal ideals:
$\down{q_1} = \{q_1\}$, $\down{q_2} = \{q_2\}$, $\down{q_3} = \{q_2, q_3\}$,
$\down{q_4} = \{q_1, q_4\}$, $\down{q_5} = \{q_1, q_4, q_5\}$.
Sizes: $|\down{q_1}| = |\down{q_2}| = 1$, $|\down{q_3}| = |\down{q_4}| = 2$,
$|\down{q_5}| = 3$.

Synthetic scores are calibrated to reflect plausible relative rankings: user $u$
has an architectural history interest (high $\mathrm{Pref}$ for $q_1$, $q_4$, $q_5$)
and moderate literary heritage interest ($q_2$, $q_3$).

\subsection{Phase 1: State Assessment}
\label{sec:example_phase1}

Confirmed state $E_u(t) = \{q_1\}$. No additional engagement signals, so BLIM
leaves $P_u$ unchanged and $\hat{E}_u = \{q_1\}$.

$\mathrm{Rel}$ values w.r.t.\ $\hat{E}_u = \{q_1\}$:
\[
  \mathrm{Rel}(q_2, \{q_1\}) = 0, \quad
  \mathrm{Rel}(q_3, \{q_1\}) = 0, \quad
  \mathrm{Rel}(q_4, \{q_1\}) = \tfrac{1}{2}, \quad
  \mathrm{Rel}(q_5, \{q_1\}) = \tfrac{1}{3}.
\]

\begin{remark}[$\mathrm{Rel} = 0$ for prerequisite-free fringe items]
\label{rem:rel_zero_fringe}
When $\down{L_i} = \{L_i\}$, $\mathrm{Rel}(L_i, \hat{E}_u) = 0$ if $L_i \notin
\hat{E}_u$. For $q_2$: $\mathrm{Rel}(q_2, \{q_1\}) = 0$ yet $q_2 \in \Fringe(\{q_1\})$,
since the fringe condition $\down{q_2} \setminus \{q_2\} = \emptyset \subseteq
\{q_1\}$ holds trivially. These two facts coexist without contradiction. The $w_{\delta}
\cdot \mathrm{Rel}$ term is 0 for such items regardless of $w_{\delta}$; 
structural differentiation within the fringe for
prerequisite-free items can be provided by the depth-normalized
score $\mathrm{Depth}(L_i) = |\down{L_i}|/n$,
which is state-independent and complements $\mathrm{Rel}$.
\end{remark}

With equal weights $w_{\alpha}=w_{\beta}=w_{\gamma}=w_{\delta}=0.25$:

\begin{table}[htb]
\centering
\small
\caption{Interest score components given $\hat{E}_u = \{q_1\}$.
``excl.''\ denotes that $q_1 \in \hat{E}_u$ and is excluded from
recommendation scoring. $q_3$ and $q_5$ are labelled ``blocked''
because they are not in $\Fringe(\hat{E}_u)$ (Phase~2 fringe gate);
$q_2$ is fringe-accessible despite $\mathrm{Rel}=0$
(Remark~\ref{rem:rel_zero_fringe}).}
\label{tab:example_scores}
\begin{tabular}{lrrrrr}
\toprule
 & $q_1$ & $q_2$ & $q_3$ & $q_4$ & $q_5$ \\
\midrule
$\mathrm{Pref}_{u,i}$          & 0.90 & 0.60 & 0.50 & 0.80 & 0.85 \\
$\mathrm{Prop}_i$               & 0.85 & 0.70 & 0.65 & 0.75 & 0.80 \\
$\mathrm{Collab}_{u,i}$         & 0.80 & 0.65 & 0.55 & 0.70 & 0.80 \\
$\mathrm{Rel}(L_i, \hat{E}_u)$ & $1$ (visited) & $0$ & $0$ & $\tfrac{1}{2}$ & $\tfrac{1}{3}$ \\
\midrule
$\I(u, L_i, \hat{E}_u)$ & excl. & $0.4875$ & blocked & $0.6875$ & blocked \\
\bottomrule
\end{tabular}
\end{table}

\[
\I(u, q_2, \{q_1\}) = 0.25 \times (0.60 + 0.70 + 0.65 + 0) = 0.4875, \quad
\I(u, q_4, \{q_1\}) = 0.25 \times (0.80 + 0.75 + 0.70 + 0.5) = 0.6875.
\]

\subsection{Phase 2: Fringe Construction}
\label{sec:example_phase2}

Applying Definition~\ref{def:fringe} to $K = \{q_1\}$:
\begin{itemize}[noitemsep]
  \item $q_2$: $\down{q_2} \setminus \{q_2\} = \emptyset \subseteq \{q_1\}$ \checkmark
  \item $q_3$: $\down{q_3} \setminus \{q_3\} = \{q_2\} \not\subseteq \{q_1\}$ \xmark
  \item $q_4$: $\down{q_4} \setminus \{q_4\} = \{q_1\} \subseteq \{q_1\}$ \checkmark
  \item $q_5$: $\down{q_5} \setminus \{q_5\} = \{q_1,q_4\} \not\subseteq \{q_1\}$ \xmark
\end{itemize}
Hence $F = \{q_2, q_4\}$. Augmented states $\{q_1,q_2\}$ and $\{q_1,q_4\}$ are
both valid order ideals (verified by checking downward closure), confirming
Lemma~\ref{lem:fringe_valid}.

\subsection{Phase 3: Path Recommendation via Memoized DP ($k_{\max} = 2$)}
\label{sec:example_phase3}

The memoized recursion starts with \textsc{DPVal}$(\{q_1\}, 2)$.
This triggers two recursive calls,
\textsc{DPVal}$(\{q_1,q_2\}, 1)$ and
\textsc{DPVal}$(\{q_1,q_4\}, 1)$, before resolving the root.
At horizon $k=1$:
\begin{itemize}[noitemsep]
  \item $\Fringe(\{q_1,q_2\}) = \{q_3, q_4\}$ (verified:
        $\down{q_3}\setminus\{q_3\}=\{q_2\}\subseteq\{q_1,q_2\}$ \checkmark;
        $\down{q_4}\setminus\{q_4\}=\{q_1\}\subseteq\{q_1,q_2\}$ \checkmark;
        $\down{q_5}\setminus\{q_5\}=\{q_1,q_4\}\not\subseteq\{q_1,q_2\}$ \xmark)
  \item $\Fringe(\{q_1,q_4\}) = \{q_2, q_5\}$ (verified similarly)
\end{itemize}

Interest scores at intermediate states, computed via Eq.~\eqref{eq:rel} (fringe
items satisfy $\mathrm{Rel} = (|\down{L_i}|-1)/|\down{L_i}|$, not 1):
\begin{alignat*}{2}
\mathrm{Rel}(q_5,\{q_1,q_4\}) &= \tfrac{|\{q_1,q_4,q_5\}|-1}{|\{q_1,q_4,q_5\}|}
  = \tfrac{2}{3}, &\quad
\mathrm{Rel}(q_2,\{q_1,q_4\}) &= \tfrac{|\{q_2\}|-1}{|\{q_2\}|} = 0, \\
\mathrm{Rel}(q_4,\{q_1,q_2\}) &= \tfrac{|\{q_1,q_4\}|-1}{|\{q_1,q_4\}|}
  = \tfrac{1}{2}, &\quad
\mathrm{Rel}(q_3,\{q_1,q_2\}) &= \tfrac{|\{q_2,q_3\}|-1}{|\{q_2,q_3\}|}
  = \tfrac{1}{2}.
\end{alignat*}
\begin{alignat*}{2}
\I(u, q_5, \{q_1,q_4\}) &= 0.25\times(0.85+0.80+0.80+\tfrac{2}{3})
  \approx 0.7792, &\quad
\I(u, q_2, \{q_1,q_4\}) &= 0.25\times(0.60+0.70+0.65+0)
  = 0.4875, \\
\I(u, q_4, \{q_1,q_2\}) &= 0.25\times(0.80+0.75+0.70+\tfrac{1}{2})
  = 0.6875, &\quad
\I(u, q_3, \{q_1,q_2\}) &= 0.25\times(0.50+0.65+0.55+\tfrac{1}{2})
  = 0.5500.
\end{alignat*}

Memo table (base cases zero; memoization avoids recomputing $\{q_1,q_2,q_4\}$,
reached by two distinct length-2 paths):
\begin{center}\small
\begin{tabular}{lrc}
\toprule State & $j$ & $\mathtt{memo}$ \\ \midrule
$\{q_1,q_2,q_3\}$ & 0 & 0 \\
$\{q_1,q_2,q_4\}$ & 0 & 0 \\
$\{q_1,q_4,q_5\}$ & 0 & 0 \\
$\{q_1,q_2\}$ & 1 & $\max(0.6875,\,0.5500) = 0.6875$ \quad[$\mathtt{pred}=q_4$] \\
$\{q_1,q_4\}$ & 1 & $\max(0.7792,\,0.4875) \approx 0.7792$ \quad[$\mathtt{pred}=q_5$] \\
$\{q_1\}$ & 2 & $\max(0.4875+0.6875,\;0.6875+0.7792)
  = \max(1.1750,\,1.4667) \approx \mathbf{1.4667}$ \\
\bottomrule
\end{tabular}
\end{center}

Backtracking: $\mathtt{pred}[(\{q_1\},2)] = q_4$, then
$\mathtt{pred}[(\{q_1,q_4\},1)] = q_5$.
Optimal path: $\pi^* = (q_4, q_5)$,
$V^* = \tfrac{11}{16} + \tfrac{187}{240} = \tfrac{352}{240} = \tfrac{22}{15} \approx 1.4667$.

\begin{table}[htb]
\centering\small
\caption{Complete DP path evaluation from $\hat{E}_u = \{q_1\}$,
$k_{\max} = 2$. $\I_2$ values use exact Rel scores per Eq.~\eqref{eq:rel};
$\I(q_5,\cdot)$ involves $\mathrm{Rel}=2/3$, giving the exact
fraction $187/240 = 0.7791\overline{6}$ (a non-terminating
repeating decimal), approximated to 4~d.p.
The state $\{q_1,q_2,q_4\}$ is reached by two paths but evaluated
once (memoization). ``State valid?'' confirms
Corollary~\ref{cor:path_valid}.}
\label{tab:example_dp}
\setlength{\tabcolsep}{4pt}
\begin{tabular}{lllrrl l}
\toprule
\textbf{Step 1} & \textbf{State} & \textbf{Step 2} & $\I_1$ & $\I_2$ & $V$ & \textbf{Valid?} \\
\midrule
$q_4$ & $\{q_1,q_4\}$, $\Fringe=\{q_2,q_5\}$ & $q_5$ & 0.6875 & $\approx$0.7792 & $\approx\mathbf{1.4667}$ & \checkmark \\
      &                                         & $q_2$ & 0.6875 & 0.4875         & 1.1750 & \checkmark \\
\midrule
$q_2$ & $\{q_1,q_2\}$, $\Fringe=\{q_3,q_4\}$ & $q_4$ & 0.4875 & 0.6875 & 1.1750 & \checkmark \\
      &                                         & $q_3$ & 0.4875 & 0.5500 & 1.0375 & \checkmark \\
\bottomrule
\end{tabular}
\end{table}

\paragraph*{Structural verification.}
Path $\pi^* = (q_4, q_5)$ traverses:
$\{q_1\} \xrightarrow{+q_4} \{q_1,q_4\} \xrightarrow{+q_5} \{q_1,q_4,q_5\}$,
each an upward edge in the Hasse diagram of $(\K, \subseteq)$, ascending two levels.
This demonstrates the well-graded property (Proposition~\ref{prop:knowledge_space}(iv)).
Sub-path optimality: $\I(u,q_5,\{q_1,q_4\}) \approx 0.7792 > 0.4875 = \I(u,q_2,\{q_1,q_4\})$
confirms $q_5$ is uniquely optimal from $K_1$, verifying Proposition~\ref{prop:subpath}.
Memoization saving: at scale (20~POIs, $k_{\max}=5$), there are
at most $\binom{20}{5} = 15{,}504$ distinct size-5 state sets,
each reached by up to $5! = 120$ orderings of the same items.
Memoization collapses all orderings reaching the same state into
a single computation, giving a factor-of-$120$ reduction in
evaluations before any structural pruning from $\preceq$.

\subsection{Feedback Processing}
\label{sec:example_feedback}

User $u$ visits $q_4$ with high-confidence signal ($I_{\mathrm{obs}} = 0.9$).
Algorithm~\ref{alg:feedback} executes:

\textbf{Step 1 (preference):}
$p_{u,4}(t') = 0.80 + 0.1 \times (0.9 - 0.80) = 0.81$.

\textbf{Step 2 (BLIM):} With $\beta_4 = 0.05$, $\eta_4 = 0.10$, states containing
$q_4$ receive prior $\times 0.95$; others receive prior $\times 0.05$. After
normalization, $P_u'$ concentrates on states containing $q_4$, with MAP
$\hat{E}_u' = \{q_1, q_4\}$.

\textbf{Step 3 (confirmed state):} $q_4 \in \Fringe(\{q_1\}) = \{q_2,q_4\}$ \checkmark.
$E_u(t') = \{q_1\} \cup \{q_4\} = \{q_1, q_4\} \in \K$ \checkmark.

\textbf{Incremental fringe update:} The sole direct successor of $q_4$ is $q_5$:
$\mathtt{cnt}[q_5] \leftarrow 0$, so $q_5$ enters the fringe. New fringe:
$\Fringe(\{q_1,q_4\}) = \{q_2,q_5\}$, computed in
$O(\mathrm{deg}_H^+(q_4)) = O(1)$.

The visit to $q_4$ is a structural event: the confirmed state advanced one level
in the lattice, unlocking $q_5$ and raising $\mathrm{Rel}(q_5, \hat{E}_u)$ from
$\tfrac{1}{3}$ to $\tfrac{2}{3}$ (fringe item of depth~3, per Eq.~\eqref{eq:rel}). The next pipeline cycle produces a substantively different
recommendation set from any preference-based system lacking the structural layer.

\subsection{Summary of Verified Properties}

\begin{table}[htb]
\centering\small
\caption{Formal properties verified in the worked example.}
\label{tab:example_verification}
\setlength{\tabcolsep}{5pt}
\begin{tabular}{llll}
\toprule
\textbf{Property} & \textbf{Result} & \textbf{Instance check} & \textbf{Outcome} \\
\midrule
States $\{q_1\},\{q_1,q_4\},\{q_1,q_4,q_5\} \in \K$ & Lemma~\ref{lem:fringe_valid} & $\down{q_i}\subseteq K_j$ verified & \checkmark \\
$\pi^*=(q_4,q_5)$ items distinct & Corollary~\ref{cor:path_valid}(ii) & $q_4\neq q_5$; each added before next step & \checkmark \\
Sub-path $(q_5)$ optimal from $\{q_1,q_4\}$ & Proposition~\ref{prop:subpath} & ${\approx}0.7792 > 0.4875$ & \checkmark \\
$\Fringe(\{q_1\})=\{q_2,q_4\}$ & Definition~\ref{def:fringe} & Verified by ideal check & \checkmark \\
Lattice path ascends via fringe & Prop.~\ref{prop:knowledge_space}(iv) & $\{q_1\}\to\{q_1,q_4\}\to\{q_1,q_4,q_5\}$ & \checkmark \\
Fringe update after $q_4$ confirmed & Prop.~\ref{prop:fringe_complexity} & $\mathtt{cnt}[q_5]$ decremented; $q_5$ enters fringe & \checkmark \\
$E_u(t)$ not modified in Phase 1--3 & Remark~\ref{rem:invariant} & Only Alg.~\ref{alg:feedback} modifies $E_u(t)$ & \checkmark \\
\bottomrule
\end{tabular}
\end{table}

\section{Conceptual Analysis}
\label{sec:analysis}

\subsection{Representational Capacity}
\label{sec:repres_capacity}

\begin{definition}[User Model, Item Model, Recommendation Model]
\label{def:models}
A recommender system is characterized by three representational choices.
The \emph{user model} $\mathcal{U}$ specifies the formal object representing a user.
The \emph{item model} $\mathcal{L}$ specifies the formal object representing a location.
The \emph{recommendation model} $\mathcal{M}_{\mathcal{R}}$ specifies the formal
function mapping $(\mathcal{U}, \mathcal{L})$ to a recommendation.
\end{definition}

Table~\ref{tab:representational_analysis} characterizes each paradigm across these
dimensions, with two additional columns: \emph{state completeness} (whether the user
model can represent a user's structural position in a formally defined space) and
\emph{recommendation validity certificate} (whether the system can prove, not merely
estimate, that a recommendation is valid for the current state).

\begin{table}[htb]
\centering\small
\caption{Representational capacity analysis. \cmark\ = supported; $\circ$ = partially
or under restrictions; \xmark\ = not representable. ``Ordered item structure'' means
the item model encodes a partial order, not merely pairwise similarity.
``Validity certificate'' means the system can prove a recommendation is structurally valid.}
\label{tab:representational_analysis}
\setlength{\tabcolsep}{4pt}
\begin{tabular}{lcccccc}
\toprule
\textbf{Dimension}
  & \makecell{\textbf{CF/}\\\textbf{Content}}
  & \makecell{\textbf{Seq.}\\(MK/RNN)}
  & \makecell{\textbf{Seq.}\\(Transf.)}
  & \makecell{\textbf{TTDP}}
  & \makecell{\textbf{CBRS}}
  & \makecell{\textbf{ESRS}\\\textbf{(ours)}} \\
\midrule
User: preference vector         & \cmark  & $\circ$ & \cmark & \xmark & $\circ$ & \cmark \\
User: interaction history       & $\circ$ & \cmark  & \cmark & \xmark & \xmark  & $\circ$ \\
User: cumulative exploration state & \xmark & \xmark & \xmark & \xmark & \xmark & \cmark \\
User: state completeness        & \xmark  & \xmark  & \xmark & \xmark & \xmark  & \cmark \\
Item: attribute vector          & \cmark  & $\circ$ & \cmark & \cmark & \cmark  & $\circ$ \\
Item: ordered structure         & \xmark  & \xmark  & \xmark & \xmark & $\circ$ & \cmark \\
Rec.: relevance estimation      & \cmark  & \cmark  & \cmark & $\circ$ & \cmark & \cmark \\
Rec.: formal path optimization  & \xmark  & \xmark  & \xmark & \cmark & \xmark  & \cmark \\
Rec.: validity certificate      & \xmark  & \xmark  & \xmark & \xmark & $\circ$ & \cmark \\
Rec.: state-dependent scoring   & \xmark  & \xmark  & \xmark & \xmark & \xmark  & \cmark \\
\bottomrule
\end{tabular}
\end{table}

The gap from CF/sequential models to EST (no cumulative state representation) is
\emph{not resolvable by additional data or model capacity}: no amount of training
data teaches a CF or Transformer model to represent the downward-closure property
of exploration states. The converse gap (EST's collaborative signal is coarser than
a full latent factor model) is \emph{resolvable by integration}: EST can adopt
any CF-derived preference estimate as $\mathrm{Pref}_{u,i}$ without modifying
other components.

\subsection{Paradigm-by-Paradigm Structural Analysis}
\label{sec:paradigm_analysis}

\paragraph*{Collaborative Filtering and Content-Based.}
CF's user model is a preference vector encoding preference intensity and similarity,
but no relational structure over items. The question ``is the user ready for this
POI?'' is not representable. CF's latent factor model captures the collaborative
signal with far greater precision than EST's Jaccard-based approximation---but this
gap is resolvable by integration, while the representational gap is not.

\paragraph*{Sequential and Session-Based Recommendation.}
Sequential models learn $f: \mathcal{H} \to \Delta(Q)$ from observed transitions.
Even with infinite training data and perfect capacity, these models cannot reliably
recover a ground-truth surmise order $\preceq$ because: (i) prerequisite transitivity
may not manifest as pairwise co-occurrence if users always visit intermediate items;
(ii) universal prerequisites show no variance in the training data; (iii) geographic
and temporal co-occurrence confounds genuine prerequisites. EST does not compete with
sequential models on next-POI accuracy under standard conditions; it introduces an
orthogonal modeling dimension that sequential models cannot represent by design.

Sequential models excel at capturing \emph{contextual transition patterns}---session
dynamics, physical position, time of day, fatigue---that the cumulative exploration
state $E_u(t)$ does not encode. This is a genuine strength of sequential approaches
that EST does not replicate.

\paragraph*{The Tourist Trip Design Problem.}
TTDP optimizes $\sum_{L_i \in \pi} s_i$ with fixed scores; ESRS optimizes
$\sum_{j=1}^k \I(u, q_j, K_{j-1})$ where the value of visiting $L_i$ depends on
the state $K_{j-1}$. This state-dependence makes the ESRS problem a sequential
decision problem over a state space rather than an orienteering instance, requiring
the Bellman principle (Proposition~\ref{prop:subpath}) for tractability. 
Formally:
TTDP's path optimization objective is a special case of the ESRS DP component
(Eq.~\eqref{eq:bellman}) when $\I$ is constant in $K$ and $\preceq$
contains only reflexive pairs (i.e., $\K(Q,\preceq) = 2^Q$); this equivalence
applies to the optimization structure only and does not extend to the BLIM
state estimator, feedback loop, or cold-start components.
ESRS augmented with time constraints (Remark~\ref{rem:ttdp_dp}) is a strictly
richer model than TTDP.

\paragraph*{Constraint-Based Recommender Systems.}
CBRS models a static requirement profile; EST models a dynamic lattice state.
CBRS is a one-shot filtering paradigm; EST is a multi-step state-transition paradigm.
Two users with identical current interest profiles but different visit histories
receive structurally different recommendation sets in EST but identical ones in CBRS.
CBRS can in principle encode an ordering constraint as a logical rule, but the
resulting system lacks the lattice-theoretic structure enabling EST's algorithmic
guarantees (CBRS earns $\circ$ in Table~\ref{tab:bg_positioning}).

\subsection{A Formal Theory of EST Explanations}
\label{sec:explainability}

\begin{definition}[Structural Explanation]
\label{def:explanation}
A \emph{structural explanation} for the recommendation of POI $L_i$ to user $u$ in
state $\hat{E}_u$ is a pair $(\pi_{\mathrm{pre}}, \sigma)$ where:
$\pi_{\mathrm{pre}}$ is a path in $H$ through $\down{L_i} \setminus \{L_i\}$
passing only through items in $\hat{E}_u$; and
$\sigma$ maps each covering relation in $\pi_{\mathrm{pre}}$ to a natural-language
justification stored in the Location Database.
\end{definition}

\begin{proposition}[Every ESRS Recommendation Admits a Structural Explanation]
\label{prop:explanation_exists}
For any $\hat{E}_u \in \K$ and $L_i \in \Fringe(\hat{E}_u)$, there exists a
structural explanation $(\pi_{\mathrm{pre}}, \sigma)$ for the recommendation of $L_i$.
\end{proposition}

\begin{myproof}
Since $L_i \in \Fringe(\hat{E}_u)$, we have $\down{L_i} \setminus \{L_i\} \subseteq
\hat{E}_u$. Two cases arise.
\textit{Case 1} ($\down{L_i} \setminus \{L_i\} = \emptyset$, no prerequisites):
set $\pi_{\mathrm{pre}} = ()$ (the empty chain) and let $\sigma$ map the empty
set of covering relations to the natural-language justification
``$L_i$ is accessible to any user without prior visits.''
The pair $((), \sigma)$ is a valid structural explanation.
\textit{Case 2} ($\down{L_i} \setminus \{L_i\} \neq \emptyset$):
any maximal chain in the finite partial order
$(\down{L_i} \setminus \{L_i\}, \preceq)$ constitutes $\pi_{\mathrm{pre}}$
(non-empty by hypothesis), and its covering relations can be annotated with
semantic justifications from the Location Database.
In both cases the pair $(\pi_{\mathrm{pre}}, \sigma)$ exists.
\end{myproof}

\begin{example}[Five-POI Instance]
When ESRS recommends $q_4$ to a user in state $\{q_1\}$:
$\pi_{\mathrm{pre}} = (q_1)$; $\sigma(q_1 \lessdot q_4)$: ``The Art Gallery's
collection is curated around artefacts whose historical provenance was established
in the medieval period documented in the City Museum. Prior engagement with the City
Museum provides the contextual grounding to understand the Gallery's curatorial logic.''
For subsequently recommending $q_5$ from $\{q_1,q_4\}$:
$\pi_{\mathrm{pre}} = (q_1, q_4)$; $\sigma(q_4 \lessdot q_5)$: ``The Rooftop Bar's
panoramic view synthesizes urban perspectives from the preceding cultural visits.''
\end{example}

\begin{remark}[Distinction from Post-Hoc Explainability]
\label{rem:explainability_distinction}
EST's structural explanations differ from post-hoc methods (LIME \citep{ribeiro2016lime},
SHAP \citep{lundberg2017unified}) in two fundamental ways. First,
Proposition~\ref{prop:explanation_exists} proves existence unconditionally; post-hoc
methods compute approximations with no existence guarantee. Second, the structural
explanation \emph{is} the recommendation process---the chain $\pi_{\mathrm{pre}}$
is the actual sequence of prerequisite constraints determining fringe membership.
No gap exists between explanation and model \citep{rudin2019stop}.
\end{remark}

\subsection{A Taxonomy of EST's Properties}
\label{sec:property_taxonomy}

\subsubsection{Class I: Structural Guarantees}

Unconditional consequences of the mathematical results in Section~\ref{sec:foundations},
holding for any surmise relation on any finite POI set:
\begin{enumerate}[label=(SG\arabic*), leftmargin=3em, noitemsep]
  \item \textbf{Lattice structure} (Propositions~\ref{prop:knowledge_space}, \ref{prop:birkhoff})
  \item \textbf{Fringe validity} (Lemma~\ref{lem:fringe_valid})
  \item \textbf{Path validity} (Corollary~\ref{cor:path_valid})
  \item \textbf{Cold-start entry points} (Proposition~\ref{prop:coldstart_guarantee})
  \item \textbf{Fringe efficiency} (Proposition~\ref{prop:fringe_complexity})
  \item \textbf{DP sub-path optimality} (Proposition~\ref{prop:subpath})
  \item \textbf{Explanation existence} (Proposition~\ref{prop:explanation_exists})
\end{enumerate}

\subsubsection{Class II: Design Properties}

Hold when architectural choices are correctly implemented and preconditions satisfied:
\begin{enumerate}[label=(DP\arabic*), leftmargin=3em, noitemsep]
  \item \textbf{Confirmed state integrity}: $E_u(t)$ always satisfies downward closure
        \emph{iff} Algorithm~\ref{alg:feedback} is the sole modifier and the fringe
        guard is correctly enforced.
  \item \textbf{Cold-start robustness} (Proposition~\ref{prop:collab_coldstart}):
        $\mathrm{Collab}_{u,i}$ is well-defined \emph{iff} $\mathbf{p}_u(t_0) \neq \mathbf{0}$.
  \item \textbf{Beam MAP faithfulness} (Remark~\ref{rem:beam_validity}): exact
        \emph{iff} the true MAP state is in the beam.
  \item \textbf{Surmise acyclicity}: $\preceq$ is a valid partial order \emph{iff}
        Algorithm~\ref{alg:prefixspan} correctly resolves all cycles.
  \item \textbf{Interest score normalization}: $\I \in [0,1]$ \emph{iff} all components
        are normalized and weights sum to 1.
\end{enumerate}

\subsubsection{Class III: Empirical Hypotheses}

Require empirical validation on real data:
\begin{enumerate}[label=(EH\arabic*), leftmargin=3em, noitemsep]
  \item \textbf{Perceived coherence}: ESRS users report higher journey coherence than
        CF/sequential baselines (controlled user study).
  \item \textbf{Exploration depth}: ESRS users achieve higher coverage $|E_u(t)|/|Q|$
        over sessions (longitudinal study).
  \item \textbf{Next-POI accuracy}: ESRS achieves competitive Recall@$K$ and NDCG@$K$
        relative to SASRec, GETNext, BERT4Rec (offline benchmark evaluation, explicitly
        deferred).
  \item \textbf{Cold-start quality}: structural cold-start outperforms popularity-based
        and demographic filtering on first-session quality (A/B test).
  \item \textbf{Surmise validity}: Algorithm~\ref{alg:prefixspan} with expert validation
        produces relations judged correct by domain experts for $\geq 80\%$ of pairs
        (annotation study).
\end{enumerate}

\textbf{Note on EH3}: ESRS is not designed to maximize next-POI accuracy on standard
benchmarks but to maximize semantic coherence. Lower Recall@$K$ than GETNext with
higher exploration coherence would be scientifically significant rather than
damaging---confirming that EST and sequential models optimize different objectives.

\subsection{Conditions for EST Advantage and Limitation}
\label{sec:conditions}

\paragraph*{Conditions favoring EST.}
(1) \emph{Semantically rich domain}: the city contains genuine prerequisite
structure. (2) \emph{Users in unfamiliar environments}: users who cannot self-select
appropriate entry points based on existing familiarity benefit most from the fringe
structure. (3) \emph{Multi-session exploration}: the cumulative state compounds value
across visits. (4) \emph{Sparse interaction data}: structural cold-start requires no
historical data.

\paragraph*{Conditions limiting EST.}
(1) \emph{Invalid surmise relation}: the quality floor is set by the quality of
$\preceq$. (2) \emph{Shallow states}: near-empty states make the fringe nearly the
full prerequisite-free item set, providing little guidance. (3) \emph{High-density
mobility}: rapid accumulation may make the fringe non-binding.
(4) \emph{Order-independent domains}: an empty surmise relation reduces EST to an
unconstrained recommender.

\subsection{Complementarity and Integration Opportunities}
\label{sec:complementarity}

EST and existing paradigms are more complementary than competitive. Four integration
architectures exploit this:

\paragraph*{EST + Transformer.}
Apply a fringe mask to the Transformer output:
$a_{ij}' = a_{ij} \cdot \mathbb{1}[L_i \in \Fringe(E_u(t))]$.
This preserves the Transformer's contextual transition patterns while guaranteeing
prerequisite validity (SG2, SG3).

\paragraph*{EST + GNN.}
A GNN trained on the POI attribute graph (category hierarchy, geographic proximity,
semantic embeddings) can predict surmise edges, addressing the trajectory-data
dependence of Algorithm~\ref{alg:prefixspan} and enabling zero-shot surmise relation
initialization for new cities.

\paragraph*{EST + CF.}
Replace the Jaccard-based collaborative signal with a CF model trained on
$(u, i, E_u(t))$ triples, learning that affinity for $L_i$ changes as the
exploration state evolves.

\paragraph*{EST + TTDP.}
The augmented DP state $(K, t_{\mathrm{elapsed}})$ of Remark~\ref{rem:ttdp_dp}
combines structural prerequisite validity with TTDP time-feasibility, producing
itineraries with both guarantees simultaneously.

\subsection{Theoretical Properties: A Synthesizing View}
\label{sec:theory_synthesis}

EST's mathematical structure forms a three-level hierarchy:

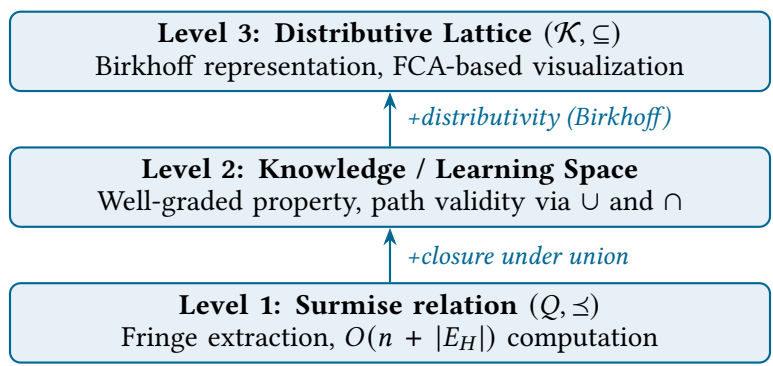
\begin{center}
\begin{tikzpicture}[
  node distance=1.8cm, 
  layer/.style={
    draw=MidnightBlue, thick, rounded corners=4pt,
    minimum width=10cm, minimum height=1cm,
    text width=9.5cm, align=center,
    font=\small\linespread{0.9}\selectfont, 
    fill=MidnightBlue!6
  },
  arrow/.style={-Stealth, thick, MidnightBlue},
  lbl/.style={font=\footnotesize\itshape, right, xshift=2pt} 
]

\node[layer] (l1) {
    \textbf{Level 1: Surmise relation $(Q,\preceq)$} \\ 
    Fringe extraction, $O(n+|E_H|)$ computation
};

\node[layer] (l2) [above of=l1] {
    \textbf{Level 2: Knowledge / Learning Space} \\ 
    Well-graded property, path validity via $\cup$ and $\cap$
};

\node[layer] (l3) [above of=l2] {
    \textbf{Level 3: Distributive Lattice $(\mathcal{K}, \subseteq)$} \\ 
    Birkhoff representation, FCA-based visualization
};

\draw[arrow] (l1) -- (l2) node[midway, lbl] {+closure under union};
\draw[arrow] (l2) -- (l3) node[midway, lbl] {+distributivity (Birkhoff)};

\end{tikzpicture}
\end{center}

Level~1 delivers efficient fringe computation. Level~2 delivers path validity and
the well-graded property. Level~3 delivers the compact canonical representation
connecting EST to FCA. The full strength of EST requires all three levels.

Three open theoretical questions are the following:
(i)~the approximation ratio for beam DP --- specifically, whether
$V_{\mathrm{beam}} \geq (1-\varepsilon(B)) \cdot V^*$ holds and
how $\varepsilon(B)$ depends on the structure of $(Q,\prec)$;
(ii)~the sample complexity for surmise relation inference (PAC-learning bounds
on the number of trajectories required for reliable recovery of $\prec$);
(iii)~the sensitivity of BLIM estimation to parameter misspecification and
its propagation to MAP estimate error.
These questions are precisely statable within the formal framework developed here,
which is itself a consequence of the explicit algebraic structure.

\section{Discussion}
\label{sec:discussion}

\subsection{Implications}
\label{sec:implications}

\subsubsection{For Recommender Systems Research}

EST's most significant implication is not a better algorithm for an existing task
but the introduction of a new \emph{evaluation dimension}: structural validity.
Current metrics (Recall@$K$, NDCG@$K$, MRR, Hit Rate) assess whether a system
predicts the next visited item; they do not assess whether recommended sequences
respect semantic ordering or guide users through coherent exploration narratives.
EST provides the formal apparatus for defining and measuring these properties.

Three concrete research implications follow. First, \emph{to the best of our knowledge, no existing LBRS benchmark contains surmise relation labels}: creating an annotated city trajectory dataset
where domain experts have labeled prerequisite relationships would enable systematic
empirical evaluation and motivate new surmise inference algorithms. Second, EST
provides a formal language for ``coherent exploration journey''---a lattice path
never violating the fringe condition---enabling the auditing of existing systems and
datasets for structural validity. Third, the BLIM inference procedure adapts naturally
to active assessment contexts, enabling minimal probe-question strategies for user
state elicitation at session start.

\subsubsection{For Urban Computing and Tourism}

EST provides theoretical grounding for recommendation systems that pursue both
personalization and structural coherence simultaneously. In \emph{heritage tourism},
the surmise relation encodes interpretive structure built into historic districts;
EST ensures visitors accumulate contextual layers necessary for deep engagement.
In \emph{newcomer orientation}, EST can model urban literacy as a formal exploration
space, guiding new residents through cultural prerequisites for integration.

A critical caution applies: EST's structural guidance is appropriate when users
\emph{seek} structured discovery. Users who prefer free, non-linear exploration may
experience the fringe constraint as paternalistic. An implementation should expose
the structural layer as an opt-in feature and allow users to override the fringe
gate, logging such overrides as potential signals of surmise relation errors.

\subsubsection{For Adjacent Domains}

EST's mathematical foundations are domain-agnostic. Promising applications include:
\emph{museum audio guides} (exhibits as POIs, surmise relation from curatorial
prerequisites, exploration state accumulating across the visit);
\emph{professional training} (skills as POIs, prerequisite graph from program
requirements, DP formulation as multi-skill sequencer);
\emph{academic advising} (courses as POIs, institutional prerequisites as surmise
relation, ESRS DP as degree-plan optimizer);
\emph{digital humanities} (primary sources organized by historiographic dependencies,
fringe guiding researchers from foundational to specialized material).

\subsection{Limitations}
\label{sec:limitations}

\paragraph*{L1. Absence of empirical validation.}
The most significant limitation is the complete absence of empirical evaluation.
No real-world deployment has been conducted; no standard benchmark has been used.
All performance claims are classified as empirical hypotheses precisely to acknowledge
this gap. Empirical investigation is the primary direction for future work.

\paragraph*{L2. The engagement signal operationalization problem.}
The distinction between meaningful and non-meaningful visits must be operationalized
through observable signals (dwell time, check-in, rating), and the threshold
$\theta_d$ is a modeling choice without a theoretically grounded default. A long
dwell time may reflect confusion or rest rather than engaged cultural consumption.
The validity of engagement signals as proxies for meaningful visits is an empirical
question requiring instrumentation beyond standard check-in datasets.

\paragraph*{L3. Monotonicity of the exploration state.}
The confirmed state is monotonically non-decreasing; the model has no mechanism for
exploration state decay, time-discounted visit relevance, or revisit semantics. A
user who visited the City Museum ten years ago and returns with no contextual memory
is modeled identically to one who visited yesterday. Temporal decay is a significant
future research direction.

\paragraph*{L4. Validity of the downward-closure assumption.}
Users frequently engage meaningfully with a POI through prior education, informal
exposure, or group tours---without a formal visit to prerequisites. The downward-closure
invariant may systematically underestimate the structural accessibility of advanced
users. Relaxing it would require replacing the learning space model with a more
permissive structure, losing the lattice-theoretic algorithmic guarantees.

\paragraph*{L5. The surmise relation elicitation bottleneck.}
Deploying ESRS in a new city requires domain experts to validate inferred surmise
links. This validation cost---in time, expert availability, and institutional
access---constitutes a significant deployment barrier. Domain experts may also
disagree: the surmise relation reflects a particular curatorial perspective rather
than a single objective ground truth.

\paragraph*{L6. Cross-city and cross-domain transfer.}
The surmise relation is city-specific with no formal mechanism for transfer, even
when two cities share cultural or architectural structures. A tourist with deep
experience of Rome's classical heritage visiting Athens should have a higher structural
starting point than a first-time visitor, but the current framework models them
identically. Cross-city structural transfer is a significant open problem.

\paragraph*{L7. Privacy of exploration states.}
Accumulated exploration states are detailed behavioral profiles. Knowing that a user
has visited a Holocaust Memorial, an LGBTQ+ History Museum, and a specific religious
institution in sequence provides sensitive identity information. EST deployments must
implement differential privacy mechanisms, user-controlled state reset, and data
minimization policies.

\paragraph*{L8. Computational bottleneck for large exploration spaces.}
The practical threshold at which beam approximation becomes necessary---and its
empirical accuracy loss---has not been characterized on real exploration data.
Tighter complexity bounds for practically relevant structures (sparse, decomposable,
nearly-tree-structured surmise relations) are an important theoretical direction.

\subsection{Future Research Directions}
\label{sec:future_research}

\paragraph*{Priority 1: Empirical implementation and evaluation.}
Dataset construction (annotated LBRS benchmark with expert surmise relation labels);
offline evaluation of Algorithm~\ref{alg:prefixspan}'s surmise recovery accuracy;
recommendation quality benchmarking against SASRec, GETNext, STAN, BERT4Rec on
standard and EST-specific metrics (fringe validity rate, structural coverage,
explanation satisfaction); controlled user study testing EH1 and EH3.

\paragraph*{Priority 2: Theoretical gaps.}
\begin{enumerate}
    \item Approximation ratio for beam DP: establish $V_{\mathrm{beam}} \geq (1-\varepsilon(B)) \cdot V^*$.
    \item Sample complexity of surmise inference: PAC-learning bounds on trajectory requirements for reliable surmise recovery.
    \item  BLIM sensitivity: characterize propagation of parameter misspecification to
MAP estimate error and recommendation quality.
\end{enumerate}
\paragraph*{Priority 3: Surmise relation learning and transfer.}
\begin{enumerate}
    \item Causal inference from trajectory data to distinguish genuine prerequisites from geographic and temporal confounders.
    \item Cross-city transfer via joint POI embeddings across shared cultural heritage.
    \item Active surmise elicitation: minimal query strategies maximally reducing
uncertainty about $\preceq$.
\end{enumerate}
\paragraph*{Priority 4: Temporal dynamics and privacy.}
\begin{enumerate}
    \item Time-discounted exploration states via recency-weighted membership, extending to fuzzy order ideals.
    \item Privacy-preserving EST: differential privacy for state distributions, federated
inference, consent-aware state management.
\end{enumerate}
\paragraph*{Priority 5: Hybrid architectures.}
\begin{enumerate}
    \item Prove that the fringe-masked Transformer preserves validity guarantees SG2 and SG3 and evaluate its effect on next-POI accuracy.
    \item Evaluate GNN-based surmise prediction for zero-shot city initialization.
    \item Design the $(u, i, E_u(t))$ interaction matrix, specify factorization, and evaluate exploration-state-aware CF against standard CF.
\end{enumerate}
\section{Conclusion}
\label{sec:conclusion}

The prerequisite structure of urban experience is real, pervasive, and consequential:
certain visits contextually ground others, and a recommender system blind to this
structure can optimize relevance only in a degraded sense. Current location-based
systems achieve considerable sophistication within this blind spot. Exploration Space
Theory removes it by giving prerequisite structure a rigorous mathematical form.

The framework rests on a single observation: the set of locations a user has
meaningfully visited is an \emph{order ideal} of a prerequisite partial order. This
observation generates an entire mathematical apparatus. Order ideals form a
distributive lattice; Birkhoff's theorem canonically identifies this lattice with
the ideal lattice of the surmise order; the lattice is well-graded; and every
fringe-guided recommendation admits a structural explanation whose existence is
provable. The architecture follows from the mathematics.

The formal core consists of six proven results:
\begin{enumerate}[label=(\textbf{M\arabic*}), leftmargin=3em, noitemsep]
  \item \textit{(Proposition~\ref{prop:knowledge_space})} The family of all order
        ideals of any surmise partial order is a well-graded learning space.
  \item \textit{(Proposition~\ref{prop:birkhoff}, Remark~\ref{rem:birkhoff})}
        The exploration space is a finite distributive lattice canonically isomorphic
        to the ideal lattice of the surmise order via Birkhoff's theorem.
  \item \textit{(Proposition~\ref{prop:fringe_complexity})} The exploration fringe
        is computable in $O(n + |E_H|)$; incremental updates cost
        $O(\mathrm{deg}_H^+(q^*))$.
  \item \textit{(Lemma~\ref{lem:fringe_valid}, Corollary~\ref{cor:path_valid})}
        Any fringe-guided transition produces a valid exploration state; any
        fringe-guided path consists of valid states and non-repeating items.
  \item \textit{(Proposition~\ref{prop:subpath})} Under additive interest gains,
        the DP value function satisfies Bellman's sub-path optimality principle.
  \item \textit{(Proposition~\ref{prop:explanation_exists})} Every ESRS recommendation
        admits a structural explanation faithful to the model by construction.
\end{enumerate}

Building on M1--M6, the paper specifies a complete system: a memoized DP exploiting
state-identity to avoid exponential path enumeration; a BLIM state estimator with EM
parameter learning and beam approximation; a feedback loop as the sole authorized
modifier of the confirmed state; an incremental surmise inference pipeline preventing
spurious relation propagation; and three cold-start strategies, of which the structural one provides
a formal validity guarantee (see \S\ref{sec:coldstart} and Limitation~L5).

EST's limitations are documented and its empirical hypotheses are clearly distinguished
from proven structural guarantees. The absence of empirical validation is the primary
limitation and the primary direction for future work. The prerequisite structure of
urban experience is not a statistical artifact that more data will eventually surface;
it is a semantic property that must be represented explicitly. Knowledge Space Theory,
developed over four decades to model exactly this kind of structured prerequisite
dependency, provides the mathematical language. The translation to urban exploration
introduces new definitions, new proofs, and new algorithms. The empirical question
of whether these tools---deployed in a real system with a real surmise relation---
produce recommendations that feel coherent and meaningful to real users remains open.
It is, we believe, a question worth asking precisely. The theoretical foundations
laid here make that empirical investigation possible.

\small
\bibliographystyle{apalike}
\bibliography{references}

\end{document}